\documentclass[12pt]{article}
\usepackage[T1]{fontenc}

\usepackage{graphicx,epstopdf,float,multirow,booktabs}
\usepackage{amsthm,amsmath,amssymb,amsfonts}
\usepackage{bbm,bm,bbold}

\usepackage{float}
\usepackage{xcolor}
\usepackage{enumitem}
\usepackage{subcaption}
\usepackage{longtable}
\usepackage{url}
\usepackage{natbib}
\usepackage{authblk}

\usepackage{tikz}
\usepackage{xcolor}
\usetikzlibrary{trees,patterns}

\usepackage{abstract}

\usepackage{lineno}
\usepackage{mathtools}
\usepackage{caption}
\usepackage{epstopdf}
\usepackage{setspace}
\newtheorem{prop}{Proposition}
\newtheorem{theorem}{Theorem}[section]
\newtheorem{definition}{Definition}[section]

\usepackage{hyperref}
\hypersetup{
    colorlinks = true,
    linkbordercolor={white},
    citecolor = {blue},
}
\usepackage{cleveref}

\usepackage[margin=1in]{geometry}

\DeclareMathOperator{\logit}{logit}
\newcommand*\diff{\mathop{}\!\mathrm{d}}

\DeclareRobustCommand{\annuity}[1]{
	\hbox{\vtop{\vbox{%
				\hrule\kern 1pt\hbox{%
					$\scriptstyle #1$%
					\kern 1pt}}\kern1pt}%
		\vrule\kern1pt}}

\begin{document}

\title{A new paradigm of mortality modeling via individual vitality dynamics}

\author[1]{Xiaobai Zhu}
\author[2]{Kenneth Q. Zhou\thanks{Corresponding author. E-mail: \textit{kenneth.zhou@asu.edu}}}
\author[1]{Zijia Wang}

\affil[1]{The Chinese University of Hong Kong, Hong Kong, China}
\affil[2]{Arizona State University, Tempe, USA}

\date{\today}

\maketitle

\begin{abstract}
    The significance of mortality modeling extends across multiple research areas, ranging from life insurance valuation to optimal lifetime decision-making. Existing approaches, such as mortality laws and factor-based models, often fall short in capturing the complexity of individual mortality, hindering their ability to address specific research needs. To overcome these limitations, this paper introduces a novel approach to mortality modeling centered on the dynamics of individual vitality. A four-component framework is developed to account for initial conditions, natural aging processes, stochastic fluctuations, and accidental events over an individual’s lifetime. We demonstrate the framework’s analytical capabilities across various settings and explore its practical implications in solving life insurance problems and deriving optimal lifetime decisions. Our results show that the proposed framework not only encompasses existing mortality models but also provides individualized mortality outcomes and offers an intuitive explanation for survival biases.
\end{abstract}


\section{Introduction}

Mortality modeling serves as a cornerstone for numerous research domains and practical applications, ranging from life insurance valuation and longevity risk management to optimal lifetime decision-making and public health policies. Accurate mortality models are essential for life insurers when evaluating and managing the demographic risks associated with life insurance products. These models are used to calculate prices, determine reserves, and fulfill regulatory requirements, ultimately facilitating better planning and management of life insurance products.

Mortality modeling also plays an important role in optimal consumption and investment strategies. By estimating the survival time of individuals, mortality models can help develop optimal strategies for better financial planning and resource allocation. Additionally, the ability to predict longevity and mortality trends is crucial in pension fund management, public health planning, and demographic studies. In each of these areas, the reliability and interpretability of mortality modeling outcomes directly impact the quality of the decisions made.

Traditional mortality modeling approaches, such as mortality laws and factor-based models, have been instrumental in performing the aforementioned tasks. However, they often fall short in capturing the intricate dynamics of individual mortality, particularly when dealing with the heterogeneity in health conditions and accidental events over a lifetime. This gap in modeling capabilities highlights the need for a more individualized approach to better account for the variability and complexity of individual mortality.

This paper aims to address these limitations by introducing a novel mortality modeling framework centered on individual vitality dynamics. This framework conceptualizes mortality as the depletion of an individual’s vitality to zero and models the time of death as a first passage problem. By integrating multiple components into the underlying vitality process, our modeling approach delivers a more nuanced and intuitive understanding of mortality dynamics, which is important for improving actuarial practices and informing optimal decisions.

\subsection{Literature review}
Mortality laws are among the oldest and most fundamental approaches in mortality modeling. These laws, such as the Gompertz-Makeham law \citep{gompertz1825, makeham1860law}, aim to describe age-specific mortality rates using simple mathematical functions. Mortality laws are widely appreciated for their simplicity and ability to capture the general trend of increasing mortality with age, but they often fall short in capturing the complexities of real-life mortality data. More recent studies on this approach include, for example, \cite{gavrilov2019new}, \cite{li2021gompertz}, and \cite{di2023closed}.

In contrast to mortality laws, another category is developed from using stochastic processes to describe mortality intensities. The seminal work by \cite{milevsky2001mortality} applied a continuous-time Cox model to survival analysis. This modeling approach preserves analytical tractability and provides closed-form solutions to survival probabilities. Other studies include, for example, \cite{biffis2005affine}, \cite{bauer2010pricing}, \cite{blackburn2013consistent}, \cite{wong2017managing}, and \cite{zhou2022stochastic}.

Factor-based models, led by the seminal works of \cite{lee1992modeling} and \cite{cairns2006two}, have become prominent due to their ability to capture the trend and volatility in real mortality data through a combination of age-specific components and time-varying factors. Significant advancements have been made in this direction, including the works of \cite{haberman2012parametric}, \cite{kleinow2015common}, and \cite{hunt2023mortality}, with a comprehensive review provided by \cite{cairns2009quantitative}.

Similarly, spline-based models offer a non-parametric approach to mortality modeling by fitting smooth spline curves to capture the intricate patterns of mortality data. \cite{currie2004smoothing} demonstrated that the use of P-splines for smoothing mortality data provides a robust method for handling irregularities and fluctuations. Other studies include, for example, \cite{dodd2018smoothing}, \cite{camarda2019smooth}, \cite{dodd2021stochastic}, and \cite{zhu2023smooth}.

\subsection{Motivation}
Recently, a new approach of modeling mortality has been proposed by \cite{shimizu2021does} and \cite{shimizu2023survival}. Specifically, \cite{shimizu2021does} introduces the Survival Energy Model (SEM), where each individual is assumed to have a survival energy level that depletes stochastically over time through a time-inhomogeneous diffusion process. \cite{shimizu2023survival} expanded the SEM framework by using an inverse Gaussian process with an explicit mortality function. By focusing on the stochastic nature of an individual's survival energy, the SEM provides an intuitive way of understanding mortality dynamics.

Within the SEM framework, an individual's death time is essentially the first passage time that the survival energy process crosses zero. This modeling approach is closely connected to the field of first passage analysis, a well-established topic in applied probability with extensive applications in operations management \citep[see, e.g.,][]{yamazaki2017inventory,han2024production}, quantitative finance \citep[see, e.g.,][]{davydov2001pricing,cai2011option}, and actuarial science \citep[see, e.g.,][]{lin2000moments,asmussen2010ruin}. For instance, the depletion of an individual’s survival energy is similar to the dynamics of an insurer’s cash flow, where a negative jump in survival energy can be equivalently viewed as an insurer’s surplus process taking a spectrally negative form.

The concept of modeling individual survival energy, also known as \emph{vitality}, is not new beyond financial literature. It was pioneered by \cite{strehler1960general}, who defined vitality as ``the capacity of an individual organism to stay alive.'' \cite{anderson1992vitality} further contributed by using a Brownian motion to describe vitality dynamics. Along this direction, numerous studies have been conducted, including the works by \cite{anderson2000vitality}, \cite{aalen2004survival}, \cite{li2009vitality}, \cite{skiadas2010comparing}, \cite{sharrow2016quantifying}, and \cite{anderson2017insights}. A discussion on the application of the Strehler-Mildvan model to mortality modeling is available in \cite{finkelstein2012discussing}.

\subsection{Our contribution}
The objective of this paper is to explore a novel approach to mortality modeling through the lens of individual vitality dynamics. By conceptualizing death as the depletion of an individual’s vitality to zero, we aim to derive mortality modeling outcomes that are individualized for unique health conditions, lifestyle factors, and random life events. We also attempt to apply the vitality-based modeling approach to a variety of emerging research problems. Specifically, this paper makes three significant contributions.

The first contribution of this paper is the introduction of a vitality-based framework for mortality modeling. This framework assumes that each individual has an initial vitality that depletes to zero over time due to both natural causes and external factors. It comprises four components: initial vitality level, natural depletion trends, stochastic diffusion, and sudden jumps. Our modeling approach offers a comprehensive and flexible method for capturing various aspects of mortality dynamics.

The second contribution lies in the various specifications of the proposed vitality-based framework. We demonstrate that certain specifications can represent existing mortality models, while others can generate more diverse modeling outcomes. Under each specification, we derive survival probabilities and other relevant properties to illustrate how the vitality-based framework can complement and improve existing models in describing complex mortality dynamics at the individual level.

The third contribution of this paper relates to the practical applications of the vitality-based framework across several research domains. Specifically, we discuss how individual vitality can be incorporated into calculating life expectancies, pricing insurance products, and deriving optimal consumption decisions. Additionally, we derive disability and recovery probabilities, analyze different causes of death, and explain the under- and overestimation of survival probabilities under our vitality-based framework.

The structure of this paper is organized as follows. Section \ref{sec:vitality_model} details the proposed vitality-based mortality modeling framework. Section \ref{sec:specifications} covers the various specifications of the proposed framework. Section \ref{sec:applications} discusses how the vitality-based framework can be incorporated in various applications. Section \ref{sec:estimation} presents an illustration on the estimation of the proposed framework. Finally, Section \ref{sec:Conclusion} concludes the paper with a discussion on limitations and future research.

\section{The Proposed Framework} \label{sec:vitality_model}

In this section, we introduce the vitality-based mortality modeling framework along with its fundamental structure and components. Consider an individual who is age $x$ at time $0$ (the time of modeling), and let $V(t)$ be the vitality of this individual at time $t \geq 0$. The individual is considered deceased when the vitality level is exhausted (i.e., when $V(t) = 0$). Let $\tau$ be the first passage time of the vitality $V(t)$ crossing zero; that is,
\[
\tau = \inf\left\{t\geq 0: V(t)\le 0\right\}.
\]
Then, the $T$-year survival probability for the individual is $\Pr(\tau > T)$, for which we aim to derive an expression in this paper.

Our modeling framework describes the dynamics of $V(t)$ as
\begin{equation} \label{eq:Vt}
    V(t) = V(0) - Y(t) - W(t) - J(t), \quad t\ge 0,
\end{equation}
where $V(0)$, $Y(t)$, $W(t)$ and $J(t)$ are the four components with the following interpretations:
\begin{enumerate}
    \item[1.] $V(0)$ represents the vitality level at time $0$, also called the initial vitality.
    \item[2.] $Y(t)$ represents the trend at which the vitality level is naturally depleted.
    \item[3.] $W(t)$ represents the diffusion underlying the depleting rate of the vitality level.
    \item[4.] $J(t)$ represents the jump process that may affect the vitality level.
\end{enumerate}
In the following subsections, We describe these components in detail and discuss their conceptual underpinnings and intuitive implications.

\subsubsection*{Component 1 (C1): The initial $V(0)$}
The first component of our vitality-based modeling framework is the initial vitality of the individual being modeled, denoted as $V(0)$. This component sets the baseline from which all future vitality changes are measured. If the individual's age $x$ is zero, then $V(0)$ represents the vitality level of a newborn. Otherwise, $V(0)$ can be interpreted as the remaining vitality level of the individual $x$ years after birth.

We can model $V(0)$ either as a fixed value or a random variable. A fixed value of $V(0)$ assumes that the level of initial vitality is known and also the same for all of the individuals being modeled. This setting may be reasonable if the group of individuals is homogeneous, for example, a portfolio of selected insurance policyholders. We may also assume that all newborns from a certain birth year will have the exact same initial vitality.

In a more realistic setting, the value of $V(0)$ should be treated as random to represent the inherent variability in the initial vitality level of individuals. This randomness reflects the diversity in individuals within a group, such as genetics, lifestyles, and other factors that affect an individual's initial health condition. We would thus model $V(0)$ as a random variable with
\[
V(0) \sim F_0,
\]
where $F_0$ is a distribution that could be, for example, exponential or Pareto.

Consider a group of individuals all aged $x>0$. Their remaining vitality after $x$ years of living will differ due to their personal lifestyle, living environment, and other individual-specific factors. Even for newborns with age $x=0$, the initial vitality should vary within a cohort year to reflect different birth traits, such as genetic disorders and birth defects. The random setting of $V(0)$ is useful for constructing a model that can reflect individual differences, and also for building a linkage between existing mortality models and vitality-based models.

\subsubsection*{Component 2 (C2): The trend $Y(t)$}
The second component describes the decline in an individual's vitality due to the natural aging process. We denote $Y(t)$ to represent the depletion trend in vitality that occurs as the individual ages beyond time $0$. First, $Y(t)$ can be modeled as a deterministic function of time. In particular, we can assume that
\[
\diff Y(t) = \mu(t) \, \diff t, \quad t\ge 0,
\]
where $\mu(t)$ is the vitality depletion rate at time $t$ when the individual is aged $x+t$. The exact choice of $\mu(t)$ should be informed by both theoretical considerations and empirical studies.

The function $\mu(t)$ can be specified in various forms, depending on the specific characteristics and assumptions made. Common choices include constant, exponential, or more complicated functional forms that better fit empirical observations. For instance, an exponential decline might be used to represent a situation where the rate of vitality loss accelerates with age. Another example would be using the Gompertz law to specify $\mu(t)$; that is, $\mu(t) = bc^{x+t}$ where $b$ and $c$ are the conventional Gompertz parameters.

Alternatively, we may assume that the depletion rate is uniform over time; that is, $\mu(t) = \delta$ with $\delta$ being a constant. Under this setting, the rate at which vitality depletes is independent of time $t$, which may be reasonable when the duration of $\mu(t) = \delta$ is short. Moreover, we may further assume a piece-wise function of $\mu(t)$, where the depletion rate is flat over each integer age of the individual; that is, $\mu_x(t) = \delta_{x+\lfloor t \rfloor}$. This assumption is equivalent to the constant force of mortality assumption frequently used in mortality modeling.

\subsubsection*{Component 3 (C3): The diffusion $W(t)$}
The third component introduces a certain degree of uncertainty to the vitality dynamics of the individual. From the modeling perspective, given that the trend component $Y(t)$ has a deterministic nature, it is reasonable to consider a diffusion component, denoted as $W(t)$, to reflect the fact that the vitality level of an individual can fluctuate randomly around its depletion trend. This stochastic variation will account for any vitality fluctuations that an individual may experience over time. 

One simple way to model $W(t)$ is to use a Brownian motion; that is,
\[
\diff W(t) = \sigma(t) \diff B(t), \quad t\ge 0,
\]
where $B(t)$ is a standard Brownian motion and $\sigma(t)$ is a diffusion parameter at time $t$. A more advanced stochastic process can be considered to replace $B(t)$, such as a fractional Brownian motion for capturing the long-term dependence in the individual's health status over time. Additionally, a piece-wise function of $W(t)$ can also be considered to reflect that different life periods may face different levels of vitality fluctuations.

With the diffusion component, we can simulate random trajectories of vitality over time to provide insights into the variability and uncertainty inherent in the aging process. These simulated trajectories are particularly helpful in understanding the probabilistic nature of vitality depletion and are crucial for applications such as deriving personalized optimal consumption decisions.
Moreover, we can use simulated trajectories to analyze how different individuals' vitality diverge from their deterministic trend.

\subsubsection*{Component 4 (C4): The jump $J(t)$}
Another natural phenomenon that could affect an individual's health status is the occurrence of random events, such as an automobile incident or a medical surgery. These events can cause sudden and significant changes to an individual's vitality level, but cannot be captured by the depletion trend or the stochastic diffusion. To address this issue, we let the last component of our vitality-based modeling framework to be a jump component, denoted as $J(t)$. 

To incorporate the impact of accidental events into the vitality dynamics, we apply a jump process to $J(t)$; that is,
\[
J(t) = \sum_{i=1}^{N(t)} Z_i, \quad t\ge 0,
\]
where $N(t)$ counts the number of accidental events that have occurred up to time $t$, and $Z_i$ represents the size of the vitality jump caused by the $i$-th accident. We can interpret $J(t)$ as the total amount of vitality changes due to accidental events until time $t$ with random occurrence times and severity.

A straightforward example for modeling $J(t)$ would be a compound Poisson process, where $N(t)$ follows a Poisson process and $Z_i$ follows a certain severity distribution. This modeling choice is suitable because a Poisson process naturally models the occurrence of random events, while the size of each jump can follow various distributions depending on the nature of the accidental events. For instance, $Z_i$ can be modeled by a normal distribution to represent symmetric vitality changes or an exponential distribution for severe asymmetric impacts. A positive value of $Z_i$ will lead to a vitality loss, while a negative value of $Z_i$ will indicate a positive effect on vitality, such as a successful medical intervention. In a simplified setting, we may assume that $Z_i = \infty$ to represent the event of a fatal accident, where the individual's vitality drops below zero instantaneously. 

\section{Model Specifications} \label{sec:specifications}

In this section, we present several vitality models to illustrate how the proposed framework can be tailored to achieve different mortality modeling outcomes and also to enhance existing mortality models. Figure \ref{fig:vitality_connections} outlines how the remainder of this section is structured and also highlights the connection between our vitality models and some existing models.

\begin{figure}[th!]
\begin{center}
        \begin{tikzpicture}[node distance=2cm, 
section311/.style={rectangle,draw,fill=blue!10,rounded corners=.8ex},
    section312/.style={rectangle,draw,fill=red!10,rounded corners=.8ex},
    section314/.style={rectangle,draw,fill=green!20,rounded corners=.8ex},
    section321/.style={rectangle,draw,fill=brown!30,rounded corners=.8ex},
   section322/.style={rectangle,draw,fill=yellow!20,rounded corners=.8ex},
   section323/.style={rectangle,draw,fill=orange!30,rounded corners=.8ex},
   ]
   \footnotesize
\usetikzlibrary{shapes.geometric, arrows}

\tikzstyle{startstop} = [rectangle, rounded corners, minimum width=3cm, minimum height=1cm,text centered, draw=black, fill=red!30]
\tikzstyle{process} = [rectangle, minimum width=3cm, minimum height=1cm, rounded corners, text centered, draw=black]
\tikzstyle{decision} = [diamond, minimum width=3cm, minimum height=1cm, text centered, draw=black, fill=green!30]
\tikzstyle{arrow} = [thick,->,>=stealth]

\node (vitality) [startstop, rounded corners, fill=white]   {{Vitality Mortality Model}};
\node (static) [process,  fill=white, below of=vitality, xshift=-3.5cm, node distance = 2cm] {Static Parameters};
\node (dynamic) [process, rounded corners, fill=white, right of=static, node distance = 7.5cm] {Dynamic Parameters};

\node (gompertz) [process, rounded corners, section312, below of = static, minimum width=1.5cm, text width=1.5cm] {Plateau Models};
\node (laws) [process, rounded corners, section311, left of = gompertz ,minimum width=1.5cm, text width=1.5cm] {Mortality Laws};
\node (plateau) [process, rounded corners,  section314, right of=gompertz, minimum width=1.5cm, text width=2.5cm, xshift=0.5cm] {Jump-Diffusion Models};

\node (cbd) [process, section321, below of=dynamic, xshift=-2cm, text width = 1.5cm, minimum width=1.5cm] {M5 Model};
\node (cohort) [process, section322, right of=cbd, node distance  = 2cm, text width=1.5cm, minimum width=1.5cm]{M6 Model};
\node (cohortage) [process, section323, right of=cohort, node distance  = 2cm, text width=1.5cm, minimum width=1.5cm]{M8 Model};

\draw [arrow] (vitality) -- (static);
\draw [arrow] (vitality) -- (dynamic);
\draw [arrow] (static) -- (laws);
\draw [arrow] (static) -- (gompertz);
\draw [arrow] (static) -- (plateau);

\draw [arrow] (dynamic) -- (cohort);
\draw [arrow] (dynamic) -- (cbd);
\draw [arrow] (dynamic) -- (cohortage);

 \node[draw, rectangle,draw,section311, minimum width=0.7cm,label=right:\footnotesize{Section 3.1.1}] at (-6,-0.4) {};
  \node[draw, rectangle,draw,section312, minimum width=0.7cm,label=right:\footnotesize{Section 3.1.2}] at (-6,-0.7) {};
    \node[draw, rectangle,draw,section314, minimum width=0.7cm,label=right:\footnotesize{Section 3.1.3}] at (-6,-1) {};
    
  \node[draw, rectangle,draw,section321, minimum width=0.7cm,label=right:\footnotesize{Section 3.2.1}] at (5,-0.4) {};
  \node[draw, rectangle,draw,section322, minimum width=0.7cm,label=right:\footnotesize{Section 3.2.2}] at (5,-0.7) {};
   \node[draw, rectangle,draw,section323, minimum width=0.7cm,label=right:\footnotesize{Section 3.2.3}] at (5,-1) {};

\end{tikzpicture}
\end{center}
\caption{Connections between vitality mortality models and existing mortality models.}
\label{fig:vitality_connections}
\end{figure}

\subsection{Models with static parameters}\label{sec:vitality_static}

We begin from vitality models that are built by static parameters. The first goal is to establish a connection between traditional mortality laws and vitality models. We then extend those vitality models to include diffusion and jump components and derive analytical solutions for the resulting survival probabilities.

\subsubsection{Mortality laws}\label{sec:vitality_static_traditional}
Let us consider the following specification of the proposed framework:
\begin{itemize}
    \item[C1] The initial vitality $V(0)$ follows an exponential distribution with unit scale; that is, $V(0) \sim \text{Exp}(1)$.
    \item[C2] The trend component $Y(t)$ is governed by a functional form of $\mu_x(t) > 0$, specified by a mortality law. For instance, if the mortality law is Gompertz, then we have $\mu_x(t) = b c^{x+t}$ and 
    \[
    \diff Y(t) = b c^{x+t}  \diff t, \quad t\ge 0,
    \]
    where $b $ and $c$ are the two Gompertz parameters.
    \item[C3] No diffusion component is specified (i.e., $W(t) = 0$ for all $t\ge 0$).
    \item[C4] No jump component is specified (i.e., $J(t) = 0$ for all $t\ge 0$).
\end{itemize}
Conditioning on the initial vitality $V(0)$, $V(t)$ is a non-increasing function of $t$ since $\mu_x(t) >0$. We thus have
\begin{align*}
    \operatorname{Pr}(\tau>T) = \operatorname{Pr}(V(T) > 0 ) = \Pr \left( V(0) >  \int_0^T\mu_x(t)\diff t\right).
\end{align*}
Given that $V(0) \sim \operatorname{Exp}(1)$, the $T$-year survival probability can be expressed as
\begin{equation}\label{eq:SurvProb_MortLaws}
   \Pr(\tau > T) = \exp\left(-\int_0^T \mu_x(s)\diff s \right).
\end{equation}

We remark that the survival probability provided by equation \eqref{eq:SurvProb_MortLaws} is the same as the survival probability implied by the mortality law $\mu_x(s)$ for $s \in (0,T)$. The vitality-based modeling approach is able to reformulate traditional mortality laws from describing the force of mortality $\mu_x$ to modeling the depletion speed of vitality with a random initial vitality level. A similar approach has been considered in biology research. For instance, \cite{li2009vitality} considered a normal distribution for $V(0)$, whereas \cite{aalen2004survival} considered an exponential distribution for $V(0)$ to draw a connection between hazard rate models and vitality models. 

\subsubsection{Plateau death models}\label{sec:vitality_plateau}

One drawback of traditional mortality laws is that the force of mortality at extreme old ages would increase indefinitely, which contradicts the plateau effect observed in real morality data. Our proposed framework is capable of reproducing the plateau effect by specifying the first two components as
\begin{itemize}
    \item[C1] The initial vitality $V(0)$ follows a Type-II Pareto distribution with shape parameter $\alpha$ and scale parameter $\phi$. The survival function of this Pareto distribution is
    \[
    S(v) = \left(1+ \frac{v}{\phi} \right)^{-\alpha}.
    \]
    \item[C2] The trend component $Y(t)$ is governed by a mortality law $\mu_x(t)$: 
    $ \diff Y(t) = \mu_x(t)\diff t $.

\end{itemize}
The resulting $T$-year survival probability of the above modeling settings is
\begin{align}\label{eq:SurvProb_Plateau}
    \Pr (\tau>T)  = \Pr \left(V(0) > \int_0^T \mu_x(t)\diff t \right) =\left( 1 + \frac{ \int_0^T \mu_x(t)\diff t  }{\phi} \right)^{-\alpha}.
\end{align}

In particular, if $\mu_x(t)$ follows the Gompertz law, then the resulting survival probability is known as the Gamma-Gompertz model \citep[see, e.g.,][]{vaupel1979impact, missov2013gamma}. In the Gamma-Gompertz model, a frailty term $Z$ is specified such that $\mu_x(t|Z) = Z \cdot \mu_x(t)$ and $Z$ follows a Gamma distribution with shape parameter $\alpha$ and rate parameter $\phi$. The resulting $T$-year survival probability of the Gamma-Gompertz model will coincide with equation \eqref{eq:SurvProb_Plateau}, which is derived from our vitality mortality model (see Appendix \ref{app:ProofsPlateau}). \cite{missov2013gamma} has demonstrated that the Gamma-Gompertz model reflects the plateau effect such that the implied force of mortality converges to a constant as age increases. 

An alternative specification for capturing the plateau effect is
\begin{itemize}
    \item[C1] The initial vitality follows an exponential distribution with unit scale: 
    $V(0) \sim \text{Exp}(1)$.
    \item[C2] The trend component is given by 
    $\diff Y(t) = Z \mu_x(t) \diff t$,
    where $\mu_x(t)$ is specified by a mortality law and $Z$ follows a Gamma distribution with shape parameter $\alpha$ and rate parameter $\phi$.
\end{itemize}
The resulting $T$-year survival probability of the above specification is the same as equation \eqref{eq:SurvProb_Plateau} (see Appendix \ref{app:ProofsPlateau} for more details). These two specifications have distinct interpretations. The first one describes the initial vitality level of the underlying population as being more dispersed and the reason for the death plateau is due to a larger proportion of the population exhibiting a large initial vitality level. The second one implies that the death plateau is indeed due to the heterogeneity in the depletion rate. While it is possible to find other formulations that yield the same survival probabilities, our study does not aim to be exhaustive in covering all possible formulations. We provide an additional formulation at the end of Appendix \ref{app:ProofsPlateau}.

\subsubsection{A jump-diffusion vitality model} \label{sec:vitality_jump}

The traditional mortality laws, reformulated under our vitality-based framework, suffer two clear drawbacks. First, although the trend component $Y(t)$ is able to capture the natural decay of an individual's vitality level, any accidental death that may occur randomly during the individual's lifetime is not included. Second, the natural vitality decay specified by $\mu_x(t)$ is deterministic, which means the randomness that existed in the individual's health status over time is ignored.

To introduce randomness, \cite{shimizu2021does} and \cite{shimizu2023survival} have considered applying a diffusion process and an inverse-Gaussian process, respectively, to govern an individual's vitality. However, these studies assume a deterministic initial vitality and an arbitrarily chosen trend component, which consequently fails to connect the vitality-based approach with traditional mortality laws. Moreover, one disadvantage of using inverse Gaussian processes is that their implied infinite number of small jumps may not be suitable for characterizing the decrement of vitality levels due to accidents.

Under the proposed framework, we now consider adding a Brownian motion in the diffusion component to reflect random fluctuations in vitality dynamics, and a compound Poisson process in the jump component to reflect the possibility of accidental events. In particular, a jump-diffusion vitality mortality model is specified as follows:
\begin{itemize}
    \item[C1] The initial vitality $V(0)$ follows a distribution $F_0$ with positive support: 
    $V(0) \sim F_0$.
    \item[C2] The trend component $Y(t)$ is governed by
    $\diff Y(t) = \mu_x(t) \diff t$,
    where $\mu_x(t)$ is specified by a mortality law.
    \item[C3] The diffusion component $W(t)$ is specified as
    $W(t) = \sigma B(t)$,
    where $B(t)$ is a Brownian motion and $\sigma$ is the diffusion parameter.
    \item[C4] The jump component $J(t)$ is specified as
    $J(t) =  \sum_{i=1}^{N(t)}Z_i$,
    where $N(t)$ is a Poisson process with intensity rate $\lambda(t)$, and $\{Z_i\}_{i\ge 1}$ is a sequence of independent and identically distributed (i.i.d.) positive random variables with distribution function $F_Z(\cdot)$.
\end{itemize}
The dynamics of $V(t)$ is consequently governed by
\begin{equation} \label{eq:Vt_JumpDiffusion}
    V(t) = V(0) - \int_{0}^{t} \mu_x(s) \diff s - \sigma B(t) - \sum_{i=1}^{N(t)} Z_i, \quad t\ge 0.
\end{equation} 
We assume that $V(0)$, $W(t)$, and $J(t)$ are independent of each other. The resulting $T$-year survival probability can be expressed as
\begin{equation} \label{eq:SurvProb_JumpDiffusion}
     \Pr(\tau > T) = \int_{0}^{\infty} \Pr\left( v - \sigma B(t) - \sum_{i=1}^{N(t)}Z_i \geq \int_{0}^{t} \mu_x(s) \diff s, \; \forall t\leq T \right) \diff F_0(v).
\end{equation}

We remark that deriving analytical expressions for the survival probability is generally challenging in this case, unless specific forms of $Y(t)$ are assumed such as linear functions. To elaborate further, the main difficulty lies in deriving the law of the first passage time of $B(t)$ over a moving boundary, which is a non-trivial problem in applied probability (see, e.g., \cite{salminen1988first}). Furthermore, deriving analytical results about the death time using renewal arguments may also not be applicable if the trend does not preserve the renewal structure of the vitality dynamics, for instance, when $Y(t)$ is modeled under the Gompertz law. In such cases, one may resort to numerical methods to approximate the survival probability of a jump-diffusion vitality model. In this paper, we employed the Monte Carlo method as suggested in \cite{jin2017first} to numerically evaluate the survival probability as indicated in equation \eqref{eq:SurvProb_JumpDiffusion}, and provided the details in Appendix \ref{app:MonteCarlo}.

In Section \ref{sec:vitality_plateau}, we illustrated two model specifications that equally describe the death plateau. It is not surprising that more than one specification of the vitality-based modeling framework can provide outcomes that align with mortality laws. We now formulate an alternative specification that also connects with mortality laws. The following specification is assumed:
\begin{itemize}
    \item[C1] The initial vitality $V(0)$ follows a Gompertz distribution with shape parameter $\eta$ and scale parameter being one. The distribution function of this Gompertz distribution is
    $
    F(v) = 1 - e^{-\eta \left(e^{v} -1\right)}.
    $
    \item[C2] The trend component $Y(t)$ has a linear decay at the rate of $\delta$; that is,
    $\diff Y(t) = \delta \diff t$.
\end{itemize}
We can show that the resulting $T$-year survival probability is
\begin{equation}\label{eq:SurvProb_AlterGompertz}
    \Pr(\tau > T) = \exp\left( -\eta \left(e^{T\delta} -1\right)  \right),
\end{equation}
where, to obtain the Gompertz law's expression of survival probability, the Gompertz parameters $b$ and $c$ need to be parameterised by $\eta = \frac{b  c^x}{\ln(c)}$ and $\delta = \ln(c)$.

To introduce randomness in the evolution of an individual's vitality over time, we can add the diffusion component and the jump component that were previously assumed in this subsection to formulate the jump-diffusion vitality mortality model again. The dynamics of $V(t)$ will be governed by
\begin{equation}\label{eq:Vt_AlterStochastic}
    V(t) = V(0) - \delta t - \sigma B(t) - \sum_{i=1}^{N(t)}Z_i, \quad t\ge 0.
\end{equation}
One special case worth noting is when the jump represents a fatal event and the survival probability can be analytically expressed in the following proposition.

\begin{prop} \label{eq:SurvProb_AlterStochastic}
Given the vitality model \eqref{eq:Vt_AlterStochastic} with $Z_i = \infty$, the survival probability can be analytically solved as
    \begin{align}
  \notag  \Pr(\tau > T) &= {e}^{-\int_0^T \lambda(t) \diff t} \cdot \Pr \left(V(0) - \delta t -\sigma B(t) > 0 \text{ for all }t\in [0,T] \right)\\ \nonumber \\
    &={e}^{-\int_0^T \lambda(t) \diff t} \cdot  \int_0^{\infty} \Phi\left(\frac{v - \delta T}{\sigma \sqrt{T}} \right) - \exp\left(\frac{2 \delta v}{\sigma^2} \right)\Phi\left( \frac{-v-\delta T}{\sigma\sqrt{T}}  \right) \diff F_0(v),
\end{align}
where $\Phi(\cdot)$ is the distribution function of a standard Gaussian distribution.
\end{prop}

Proposition \ref{eq:SurvProb_AlterStochastic} directly follows from the law of the first passage time of a Brownian motion with linear drift (see, e.g., \cite{siegmund1986boundary}), and it is thus omitted here. In general, conditioning on $V(0)$, $V(t)$ is a spectrally negative L\'evy process (SNLP), and the Laplace transform of $\tau$ is known explicitly. This allows for the use of either analytical or numerical methods of Laplace inversion to obtain the survival probability. To illustrate this, in Appendix \ref{alternative mortality}, we provide related results for the vitality mortality model described by equation~\eqref{eq:Vt_AlterStochastic} when the jump distribution is a mixture of exponential distributions.

We remark that it might be convenient to work with the exponential transform of a vitality mortality model, such that we define the exponentially-transformed $\tilde{V}(t)$ as 
\begin{align*}
    \tilde{V}(t) = \exp\left(V(t)\right) = \exp\left(V(0) - Y(t) - W(t) - J(t)\right), \quad t\ge 0,
\end{align*}
and the individual dies if $\tilde{V}(t)$ cross the threshold $D = 1$, $\tau = \inf \left\{t\geq 0 : \tilde{V}(t)\leq 1 \right\}$.
The threshold $D$ can be arbitrarily chosen.
If $D\neq 1$, we may scale the exponential transform $\exp(-\ln(D) \cdot V(t))$ to recover a scaled threshold of one. The exponential-transformed vitality has been applied in some studies. For example, \cite{chen2022optimal} modeled $\tilde{V}(t)$ as a Geometric Brownian Motion by setting $\tilde{V}(0)$ as a constant, $\diff Y(t) = \left(\mu -\frac{\sigma^2}{2}\right) \diff t$, $\diff W(t) = \sigma \diff B(t)$ and $J(t) = 0$. For any specification of initial vitality $V(0)$, it is not difficult to obtain the corresponding distribution function for exponentially transformed initial vitality $\tilde{V}(0)$. For instance, $V(0) \sim \text{Exp}(1)$ corresponds to $\tilde{V}(0) \sim \text{Pareto}(\text{shape} = 1, \text{location} = \text{scale} = 1)$ and $V(0) \sim \text{Gompertz}$ with shape parameter $\eta$ and scale parameter of 1 corresponds to $\tilde{V}(t) \sim \text{log-Gompertz}$, also known as the inverse Weibull distribution, with the same shape and scale parameters. 

\subsection{Models with dynamic parameters}

In this subsection, we build vitality models that use on dynamic parameters and consequently have a connection with stochastic mortality models. We begin by examining a dynamic extension of the Gompertz law, and then expand our investigation to models with age and cohort effects.

\subsubsection{A vitality model with period effect} \label{sec:vitality_dynamic_simple}

We start from a straightforward example, where the depletion rate of vitality is dynamically specified with time-varying parameters. In particular, we first construct a dynamic Gompertz law with the following specifications:
\begin{itemize}
    \item[C1] The initial vitality $V(0)$ follows an exponential distribution with unit scale: $V(0) \sim \text{Exp}(1)$.
    \item[C2] The trend component $Y(t)$ is specified by a dynamic version of the Gompertz law with time-varying parameters $b(t)$ and $c(t)$:
    \[
    \diff Y(t) = \mu_x(t) \diff t =  b(t) c(t)^{x+t} \diff t, \quad t\ge 0.
    \]
    The log transforms of $c(t)$ and $b(t)$ are governed, respectively, by
    \begin{align*}
        \diff \ln c(t) &= \mu_c \diff t + \sigma_c \diff B_c(t), \quad c(0) = c_0 > 0,\\
        \diff \ln b(t) &= \mu_b \diff t + \sigma_b \left( \sqrt{1-\rho^2}\diff B_b(t) + \rho \diff B_c(t) \right), \quad b(0) = b_0 > 0,
    \end{align*}
    where $B_c(t)$ and $B_b(t)$ are two independent Brownian motions.
\end{itemize}

For the Gompertz law, a dynamic version of $b(t)$ can be interpreted as the general level of mortality for the individual at time $t$ (i.e., the intercept of mortality curve), and $c(t)$ can be interpreted as the rate of changes in mortality when the individual ages at time $t$ (i.e., the slope of mortality curve). Under the vitality mortality model, we may alternatively interpret $b(t)$ as the general rate of vitality depletion at time $t$ and $c(t)$ as the change in depletion rate when the individual ages at time $t$.

The solution of $Y(t)$ is an integral of correlated log-normally distributed random variables
\begin{align}\label{eq:vitality_integral_GBM}  
    Y(t) &= Y(0) + \int_0^t \exp\left( \mu_Y(s) + \sigma_Y(s) B(s)  \right)\diff s ,\quad t\ge 0,
\end{align}
where $B(s)$ is a standard Brownian Motion, and
\begin{align*}
    \mu_Y(t) &= \mu_b\cdot t + (x+t)\cdot \mu_c \cdot  t  + \ln b_0 + (x+t) \cdot \ln c_0,\\
    \sigma_Y(t) &= \sigma_b^2 + (x+t)^2 \cdot \sigma_c^2  + 2 (x+t) \cdot \rho \cdot  \sigma_b \cdot  \sigma_c,  
\end{align*}
Since $Y(t)$ is a non-decreasing process, the $T$-year survival probability can be expressed as $\Pr(\tau >T) = \mathbb{E}\left[ \exp\left( - Y(T) \right) \right]$.

The dynamic Gompertz law presented above is closely related to the Cairns-Blake-Dowd (CBD) model, which describes mortality over time and age for a population of individuals. A continuous version of the CBD model (applied to $\ln \mu_{x,t}$ instead of $\logit q_{x,t}$) can be constructed as follows:
\begin{align*}
    \ln \mu(x,t)   &= \kappa_{1}(t) + \kappa_{2}(t) \cdot (x-\bar{x}),\quad t\ge 0.
\end{align*}
The $T$-year survival probability of an individual aged $x$ at time $0$ under this continuous version of the CBD model is $\mathbb{E}\left[\exp\left( - \int_0^T \mu(x+t,t)\diff t \right) \right]$,
where
\begin{align*}
    \mu(x+t,t) = \exp\left(\kappa_{1}(t) - \kappa_2(t)\cdot \bar{x}\right)\cdot \left[\exp\left(\kappa_2(t) \right)\right]^{x+t }.
\end{align*}
Clearly, if $\kappa_1(t)$ and $\kappa_2(t)$ are modeled as a two-dimensional random walk with drifts \citep[see, e.g.,][]{li2021constructing}, then the parameters in the process of $\ln b(t)$ and $\ln c(t)$ can be chosen to match with the continuous CBD model as
\begin{align*}
    b(t) = \exp\left(\kappa_{1}(t) - \kappa_2(t)\cdot \bar{x}\right), \quad c(t) = \exp\left(\kappa_2(t) \right).
\end{align*}
The resulting $T$-year survival probabilities would be the same between our vitality mortality model and the continuous version of the CBD model.

Recall from equation \eqref{eq:vitality_integral_GBM} that $Y(t)$ can be expressed as an integral of log-normal random variables. Although the distribution of $Y(t)$ cannot be explicitly determined, accurate numerical approximations can be made to evaluate the survival probabilities. For example, \cite{begin2023new} demonstrated that the integral of log-normal random variables can be well-approximated by a single log-normal random variable (see also, \cite{lo2013wkb}). The approximated moment generating function from \cite{asmussen2016laplace} is used to value the expectation $\mathbb{E}[\exp(-Y(t))]$.

Our vitality mortality model can be constructed to have the same survival probability as other stochastic mortality models, whenever the force of mortality is guaranteed to be non-negative. For instance, consider the Lee-Carter model: 
\[
\ln \mu(x,t) = \alpha(x) + \beta(x)\kappa(t),\quad t\ge 0.
\]
If $\alpha(x)$ and $\beta(x)$ are specified as continuous deterministic functions and $\kappa(t)$ is governed by a continuous random walk with drift, then a vitality mortality model can be established to match the likelihood of the Lee-Carter model. Another example is the Cox-Ingersoll-Ross (CIR) model for modeling the log force of mortality, as discussed in \cite{huang2022modelling} and \cite{ungolo2023estimation}, where the likelihood aligns with our vitality model if we maintain the same specification for the depletion rate $\mu_{x}(t)$.

While the likelihood of our vitality models can align with well-known stochastic mortality models, we believe that the proposed modeling approach offers significant potential and enhanced capabilities. An immediate observation is that the Gompertz law fails to differentiate mortality events between natural causes and external causes. When modeling the death rate for all causes, it is common to model $\kappa_1(t)$ (or $\ln b(t)$ in the vitality model) as a random walk with a negative drift, implying a stable improvement in the general mortality. However, if we decompose mortality into different causes of death, distinct historical trends will be observed. 

Figure \ref{fig:energy_accidents} displays the accidental death rates from 1999 to 2016 across different age groups, obtained from the U.S. Centers for Disease Control and Prevention. Unlike the all-cause death rate, which typically shows a downward trend, the accidental death rate exhibits an upward trend for all of the age groups considered. This finding has significant implications for accidental death insurance, in which distinguishing the effects of accidental and non-accidental deaths is crucial. Under the vitality modeling framework, we are able to address this issue by specifying the jump component (C4) with a time-varying intensity that increases over time.

\begin{figure}[ht!]
    \centering
    \includegraphics[width=0.9\linewidth]{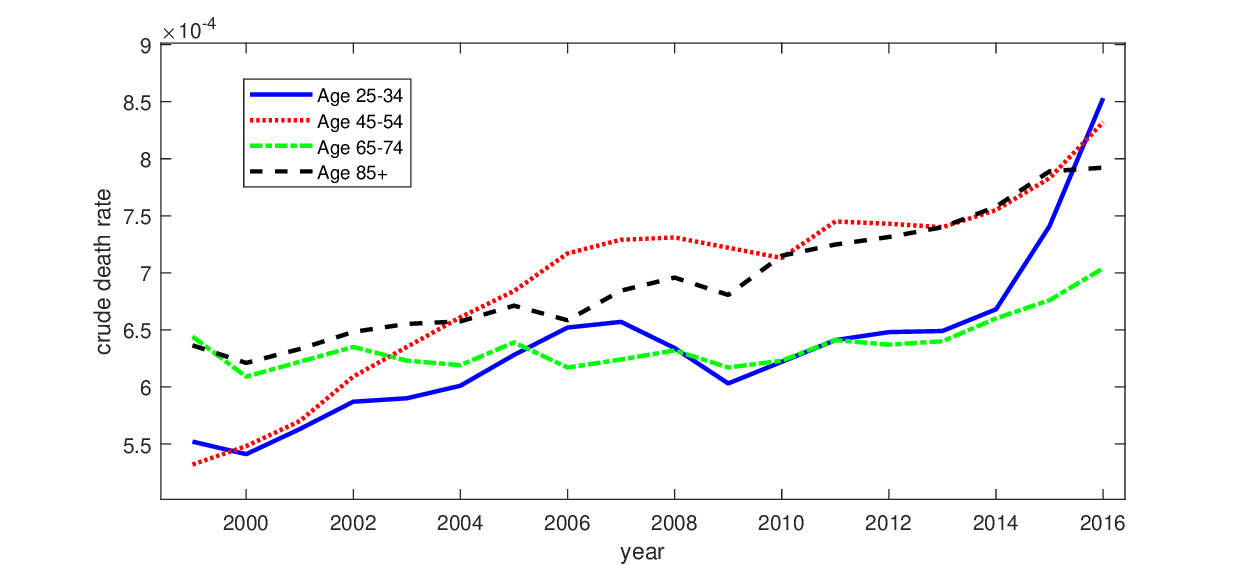}
    \caption{The observed accident death rates from 1999 to 2016 across different age groups.}
    \label{fig:energy_accidents}
\end{figure}

Considering the model setting introduced at the beginning of this section, where $Y(t)$ can be represented by an integral of log-normally distributed random variables, the $T$-year survival probability is modified by incorporating $J(t)$ which follows a compound Poisson process as defined previously:
    \begin{align*}
        \Pr(\tau>T) &= \Pr(J(T) < V(0) - Y(T))\\
        &= \int_0^{\infty} \left( \int_0^{v} \Pr( J(T)\leq v-u)  \diff F_Y(u)  \right) \diff F_0(v)
    \end{align*}
where $F_Y(\cdot)$ is the distribution functions for $Y(T)$. Since $F_0(v)$ and $\Pr(J(T)\leq v-u)$ can be explicitly expressed and $F_Y(\cdot)$ can be approximated as the distribution function of a single log-normal random variable, the integral for obtaining $\Pr(\tau>T)$ can be approximated numerically. 

\subsubsection{A vitality model with period-cohort effect} \label{sec:vitality_dynamic_cohort}

In a factor-based mortality model, the cohort effect is captured by an additive term, often denoted by $\gamma$, such as in the Renshaw-Haberman model (M2), the Age-Period-Cohort model (M3), the CBD model with cohort effect (M6), the CBD model with quadratic and cohort effects (M7), and the CBD model with age-dependent cohort effect (M8), as outlined in \cite{cairns2009quantitative}. In a vitality mortality model, the cohort effect can be naturally embedded in the initial vitality to reflect changes in the average healthiness of each generation. 

We illustrate the idea of a vitality model with cohort effect based on the dynamic Gompertz law with cohort effect, which corresponds to the CBD model with cohort effect (M6). Denote $V_y(t)$ as the time-$t$ vitality level of an individual born in year $y$ with an age of $x$ at time $0$. The following specification is assumed:
\begin{itemize}
    \item[C1] The initial vitality $V_{y}(0)$ for cohort year $y$ follows an exponential distribution with rate parameter $\Gamma(y)$, where
    \begin{equation} \label{eq:CohortProcess}
    \frac{\diff \Gamma(y)}{\Gamma(y)} = \mu_{\gamma} \diff y + \sigma_{\gamma} \diff B_{\gamma}(y), \quad \Gamma(0) = 1,
    \end{equation} 
    where $B_{\gamma}(y)$ is a standard Brownian motion.
    \item[C2] The trend component $Y(t)$ is specified as in the dynamic Gompertz model: $\diff Y(t) = b(t) c(t)^{x+t} \diff t$, $t\ge 0$.
    The log transforms of $c(t)$ and $b(t)$ are governed by
    \begin{align*}
        \diff \ln c(t) &= \mu_c \diff t + \sigma_c \diff B_c(t), \qquad c(y) = c_0(y),\\
        \diff \ln b(t) &= \mu_b \diff t + \sigma_b \left( \sqrt{1-\rho^2}\diff B_b(t) + \rho \diff B_c(t) \right), \quad b(y) = b_0(y),
    \end{align*}
    where $B_c(t)$ and $B_b(t)$ are two independent standard Brownian motions, that are also independent to $B_{\gamma}(y)$.
\end{itemize}
The time-$t$ vitality level for an individual born in year $y$ and aged $x$ at time $0$ is therefore $V_y(t) = V_y(0) - Y(t)$. The $T-$year survival probability can be expressed as
\begin{align*}
    \Pr(\tau_y > T) &= \mathbb{E}\left[ \exp\left( -\Gamma(y) Y(T) \right) \right],
\end{align*}
where $\tau_y$ is the first passage time for $V_y(t)$ crossing zero, and one can refer to Appendix \ref{dynamic vitality} for more details.

The cohort effect is reflected by the rate parameter $\Gamma(y)$, which can be intuitively interpreted as the average initial vitality level at the time of modeling for individuals born in birth year $y$. The initial condition $\Gamma(0) = 1$ ensures identifiability between $\Gamma(y)$ and $b(t)$. Note that the dynamics of $\Gamma(y)$ can be chosen to satisfy certain properties, for example, setting $\mu_{\gamma} = 0$ ensures that the future cohorts are expected to be the same as the current one, or modeling $\Gamma(y)$ as a CIR process to incorporate a mean-reverting cohort effect. 

We now establish a link between this vitality mortality model and the continuous version of the M6 model. We fix age $x$ and construct the M6 model with time $t$ and cohort $y$ as
\begin{align*}
    \ln \mu_{y}(t) &= \kappa_1(t) + \kappa_2(t) \cdot (x+t - \bar{x}) + \gamma(y) .
\end{align*}
We can again re-parameterize $b(t)$, $c(t)$, and $\Gamma(y)$ in our vitality model to match the M6 model by having
\begin{align*}
    b(t) = \exp\left(\kappa_{1}(t) - \kappa_2(t) \bar{x}\right), \quad c(t) = \exp\left(\kappa_2(t) \right), \quad \Gamma(y) = \exp(\gamma(y)).
\end{align*}
The $T$-year survival probability under the above re-parameterization of the vitality model will coincide with the M6 model.

\subsubsection{A vitality model with age-adjusted period-cohort effect}\label{sec:vitliaty_dynamic_apc}

The vitality models presented in Sections \ref{sec:vitality_dynamic_simple} and \ref{sec:vitality_dynamic_cohort} begin the modeling of vitality with a fixed initial age $x$. It would be interesting to formulate a vitality model that allows the vitality to start from an arbitrary age and the cohort effect to vary as the individual ages. To achieve this, we consider a vitality model with an age-adjusted cohort effect, extended again from the Gompertz Law. The following specification is applied:
\begin{itemize}
  \item[C1] The initial vitality $V_{x,y}(0)$ for an individual born in year $y$ and age $x$ at time $t=0$ follows
    \begin{align*}
        V_{x,y}(0) \sim \text{Exp}\left(  \Gamma(y)^{\frac{x_c - x}{x_c}} \right),
    \end{align*}
    where $x_c$ is a large and fixed value (e.g., $x_c=100$) and $\Gamma(y)$ is governed by a stochastic process (e.g., the geometric Brownian motion given by equation \eqref{eq:CohortProcess}).
    \item[C2] The trend component $Y(t)$ is specified by $\diff Y(t) = b(t) c(t)^{x+t}\diff t$, $t\ge 0$, where the dynamics of $b(t)$ and $c(t)$ are the same as those presented in Sections \ref{sec:vitality_dynamic_simple} and \ref{sec:vitality_dynamic_cohort}.
\end{itemize}

The above specification provides flexibility in modeling the interactions between cohorts and ages. More specifically, $\Gamma(y)$ captures the general level of initial vitality at birth for cohort year $y$, with the term $\frac{x_c - x}{x_c}$ reflecting an age-adjusted cohort effect as the individual ages from 0 to $x_c$. When $x=0$, the initial vitality follows $\text{Exp}(\Gamma(y))$ and the specification coincides with the one presented in Section \ref{sec:vitality_dynamic_cohort}. When $x=x_c$, the initial vitality follows $\text{Exp}(1)$ and the specification coincides with the one presented in Section \ref{sec:vitality_dynamic_simple}.

It is worth noting that, if $\Gamma(y) > 1$, the cohort effect will diminish as the individual ages, which coincides with the cohort effect observed in the M8 model.

\section{Model Applications} \label{sec:applications}

In this section, we provide a number of illustrative examples to showcase the application of a vitality mortality model in various research or practical problems. Specifically, we examine the following topics:
\begin{itemize}
    \item Pricing life annuity and insurance products based on individual vitality;
    \item Deriving optimal investment and consumption decisions using a vitality model;
    \item Explaining the biased estimation of survival rates from the vitality perspective;
    \item Modeling disability and recovery probabilities using vitality thresholds;
    \item Analyzing the cause of death using vitality components.
\end{itemize}
We believe that the vitality mortality model differs from existing mortality models not only in terms of technical details, but also by offering new interpretations and insights that cannot be explained using existing mortality models.

\subsection{Life insurance valuation}

The survival probabilities of a vitality mortality model often cannot be explicitly expressed, or are overly complicated to be used in practice. However, for the purpose of calculating life expectancies and pricing life insurance products, it is convenient to work with the death time $\tau(v)$, defined as the first passage time of vitality crossing zero for an individual with age $x$ and initial vitality $v$. 

The life expectancy of an individual with given initial vitality $v$ can be expressed as
\begin{align*}
    \overset{\circ}{e}_x(v) = \mathbb{E}[\tau(v)] = \int_0^{\infty} t f_{\tau}(t;v) \diff t,
\end{align*}
where $f_{\tau}(t;v)$ is the density function of $\tau(v)$. Using the specification from Section \ref{sec:vitality_static_traditional} (i.e., the Gompertz law), the death time is
\begin{align}\label{eq:vitality_tau_gompertz} 
    \tau(v) &:= \left\{ t \geq 0 : v- \frac{b}{\ln c}c^{x}(c^t-1) = 0 \right\} \nonumber  \\ 
    &= \frac{1}{\ln c} \ln \left(\frac{v \ln c}{b c^x} + 1 \right),
\end{align}
and the life expectancy for a population of individuals with initial vitality $V(0)\sim \text{Exp}(1)$ is 
\begin{align*}
  \overset{\circ}{e}_x = \mathbb{E} \left[ \tau(V(0)) \right] =  \int_0^{\infty} \frac{1}{\ln c} \ln   \left(\frac{v \ln c}{b c^x} + 1 \right)  e^{-v} \diff v.
\end{align*}

Similarly, the price of a life annuity with \$1 per year payable continuously to an individual with age $x$ and initial vitality $v$ is
\begin{align*}
    \bar{a}_x(v) = \mathbb{E}\left[ \bar{a}_{ \annuity{\tau(v)} } \right] = \int_0^{\infty} \bar{a}_{\annuity{t}} f_{\tau}(t;v) \diff t,
\end{align*}
where $\bar{a}_{\annuity{t}} = \frac{1-e^{-\delta t}}{\delta}$ with $\delta$ being the force of interest rate. The annuity price for a population of individuals with initial vitality $V(0)\sim \text{Exp}(1)$ is
\begin{align*}
    \bar{a}_{x} &= \mathbb{E}\left[ \bar{a}_{\annuity{\tau(V(0))}} \right] = \int_0^{\infty} e^{-v} \int_0^{\infty} \bar{a}_{\annuity{t}} f_{\tau}(t;v) \diff t \diff v.
\end{align*}
Using the relation between annuity and insurance, we immediately have $\bar{A}_x(v) = \mathbb{E}\left[e^{-\delta \tau(v)}\right] = 1 - \delta \bar{a}_x(v)$ and $\bar{A}_x =  1- \delta \bar{a}_x$.

If the density function of $\tau(v)$ can be explicitly derived, then the life expectancy, annuity price, and insurance price can also be expressed explicitly in an integral form. This is the case for our proposed model if the diffusion component is not considered. When the density function of $\tau(v)$ cannot be explicitly written, an approximation method is then needed to provide accurate estimates for life expectancy and prices. For example, consider a vitality mortality model with initial vitality, trend, and diffusion (i.e., components C1-C3), then the problem becomes finding the first exit density of Brownian motion with a curved boundary. A number of studies, including \cite{jennen1981first}, \cite{jennen1985second}, and \cite{durbin1992first}, has provided accurate and easy-to-implement approximation methods to handle this problem. \cite{skiadas2010comparing} has specifically applied approximation methods to vitality modeling. For more details on approximating the density function of the death time when C4 is absent, interested readers are referred to Appendix \ref{appendix: approximation}.

We remark that, under the vitality modeling approach, the evaluation of life expectancy and mortality-related products can be applied to specific groups of individuals, assuming that the vitality information can be obtained or estimated. For instance, by altering the distribution of initial vitality (or other components as well), the mortality of a selected healthy population or a disabled population can be analyzed. Consequently, depending on the population's health status, a more accurate price can be produced for different life insurance products. This pricing approach, which tailors the vitality mortality model to reflect specific health conditions, allows for a more personalized experience in actuarial valuation. Additionally, this approach can be extended to assess the impact of lifestyle changes, medical advancements, and other public health interventions on life expectancy and insurance prices.

\subsection{Optimal investment and consumption strategies}

One promising application of vitality modeling is in lifetime planning, e.g., Merton's portfolio problem. Individuals, when making decisions about their investment portfolio, consumption, insurance purchases, medical expenses, educational investment, etc., will take into account their health level at the time. The vitality mortality model provides a natural way to frame the optimization problem. There exist a few studies that utilize the vitality-based approach. For example, \cite{chen2022optimal} modeled vitality as a geometric Brownian motion to study the optimal insurance decision, and \cite{cheng2023reference} modeled vitality using a jump process to study the optimal medical investment. In this subsection, we provide an illustration of Merton's portfolio consumption problem under our vitality mortality model.

\subsubsection{The market and vitality processes}
Following the seminal work of \cite{merton1969lifetime}, we construct a continuous-time model for the financial market and apply the vitality model \eqref{eq:Vt_AlterStochastic} without jumps. Denote $(\Omega, \mathbb{F}, \mathbb{P})$ as a complete probability space with a right-continuous filtration generated by a two-dimensional standard Brownian Motion $(B_S(t), B_V(t))$. An individual allocates two representative assets, one risk-free $S_0(t)$ and one risky $S_1(t)$ with the following dynamics:
\begin{align*}
    \diff S_0(t) &= rS_0(t)\diff t, \quad S_0(0)>0,\\
    \diff S_1(t) &= (r+\theta \sigma_S) S_1(t)\diff t + \sigma_S S_1(t) \diff B_S(t), \quad S_1(0)>0,
\end{align*}
for $t\ge 0$, where $r$ is the risk-free rate, $\theta$ is the market price of risk, and $\sigma_S$ is the vitality of the risky asset. 

Let $\pi(t)$ be the portfolio weight invested in the risky asset, and $\zeta(t)$ be the consumption amount. Then the asset $A(t)$ of an individual has the following dynamics
\begin{align*}
    \diff A(t) = A(t) \left[ 
r + \pi(t)  \theta  \sigma_S \right]\diff t + \sigma_S  \pi(t) A(t)\diff B_S(t) - \zeta(t)\diff t, \quad A(0) = A_0.
\end{align*}

We now adopt the vitality mortality model from Eq.\eqref{eq:Vt_AlterStochastic} without jumps, which generalizes the Gompertz law if the initial vitality follows the Gompertz distribution; that is,
\begin{align*}
    \diff V(t) = -\delta \diff t + \sigma_V\diff B_V(t), \quad V(0) = V_0 > 0,
\end{align*}
for $t\ge 0$, and we assume an individual is able to track their vitality level over time.

\subsubsection{The optimal investment and consumption problem}
An individual is maximizing their lifetime utility by choosing the optimal investment and consumption strategies. Therefore, for any time $t\geq 0$, we define the value function $J(A(t)=a, V(t)=v)$ as
\begin{align}\label{eq:objective}
    J(A(t)=a, V(t)=v) &= \max_{\pi, \zeta\in \Pi} \mathbb{E}_{t,a,v} \left[\int_{t}^{\tau} e^{-\beta  (s-t)}\cdot u(\zeta(s)) \diff s + \lambda \cdot  u(A(\tau))\right],
\end{align}
where $\tau$ is the first passage time that vitality level crosses zero, $\beta>0$ is the time-preference rate, $u(\zeta) = \ln \zeta$ is the utility function (log-utility is selected for illustration), $\lambda \geq 0$ is the weight given to the utility of terminal asset (i.e., bequest), $\mathbb{E}_{t,a,v} = \mathbb{E}[\cdot|t, A(t) = a, V(t) = v]$, and $\Pi$ is the admissible set with definition given in Appendix \ref{app:optimal}. Notice that the optimization is time-independent since $V(t)$ is stationary and thus the value function does not depend on the time $t$.

The optimization problem given in equation \eqref{eq:objective} can be solved by the Hamiltonian-Jacobi-Bellman equation, as discussed in \cite{chen2022optimal}. The solution of the optimal investment strategy $\pi^*(t)$ and consumption strategy $\zeta^*(t)$ are summarized in the remark below, where the details of derivations are presented in Appendix \ref{app:optimal}.

\begin{theorem}
The optimal investment and consumption strategies of Problem \eqref{eq:objective} are
\begin{align*}
    \pi^*(t) & = \frac{ \theta  }{\sigma_S}, \quad \zeta^*(t) = \frac{A(t)}{f(V(t))}.
\end{align*}
where
\begin{align*}
    f(v) &=  \left(\lambda - \frac{1}{\beta}\right)  e^{k_1 v} + \frac{1}{\beta}, \quad \text{with} \quad  k_1 = \left\{\begin{array}{ll}
       \frac{ \delta - \sqrt{\delta^2 + 2 \sigma^2_V  \beta} }{ \sigma^2_V }<0,  & \sigma_V>0 \\
       \frac{-\beta}{\delta},  &  \sigma_V = 0
    \end{array}\right.
\end{align*}
As for comparison, we have Merton's optimal portfolio and consumption strategies presented in Appendix \ref{app:optimal}. The optimal investment strategy follows the same as Merton's optimal portfolio in an infinite time horizon, and the consumption strategy also collapsed to Merton's case (without mortality) if vitality is infinite. 
\end{theorem}

\subsubsection{Further remarks}
It is not surprising to observe that the individual's consumption depends on the vitality level, but what is interesting here is that they are not always negatively correlated, as one would expect. When the bequest motive is moderate or low (i.e., $\lambda<1/\beta$), a reasonable consumption profile is observed such that a higher vitality represents a longer time to live with, and thus a lower consumption must be made to make sure an individual will not run out of money before they are deceased. 

However, when the bequest motive is significantly large (i.e., $\lambda>1/\beta$), decreasing the vitality actually increases the consumption, a counter-intuitive result at first glance. This is due to the fact that individuals are overly concerned about their bequest, and thus leads to an overly conservative consumption initially. Once they are aging and the time for death becomes more certain, they become more aware of their over-saving and may wish to increase their consumption a bit, as it would not jeopardize their bequest goals but will significantly enhance their lifetime utility. Figure \ref{fig:vitality_consumption} illustrates the sample paths for the optimal consumption rate (as a \% of wealth) and the corresponding vitality process. The left panel represents a low bequest motive individual ($\lambda < 1/\beta$), and the right panel represents a high bequest motive individual ($\lambda>1/\beta$). It is evident that a higher bequest motive results in a significantly lower consumption profile and a positive correlation between the consumption rate and vitality level. Conversely, when the bequest motive is low, the opposite pattern is observed.

\begin{figure}[H]
    \centering
    \includegraphics[width=1\linewidth]{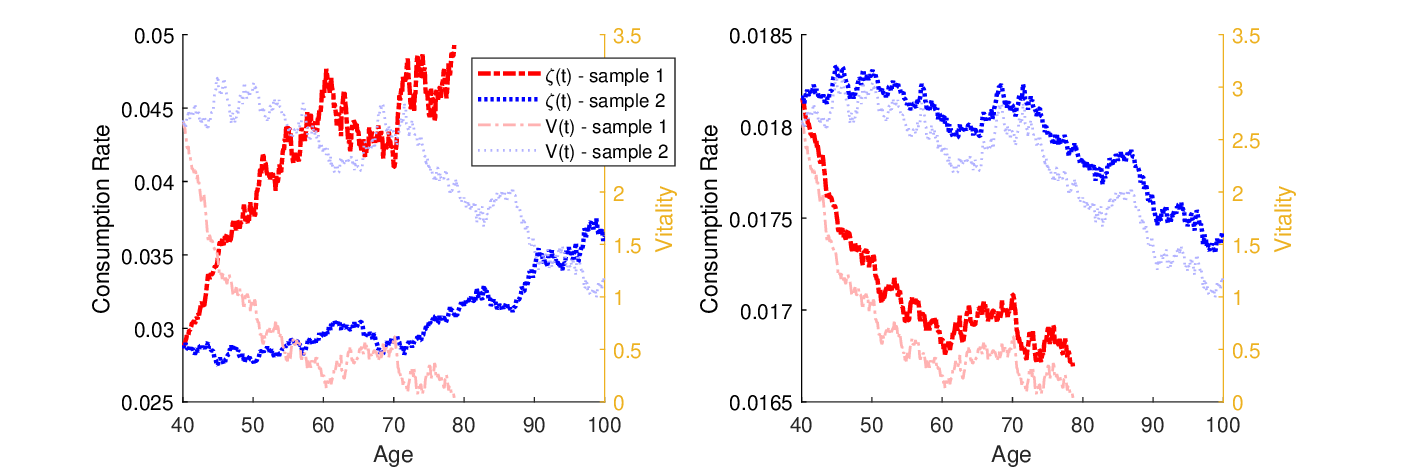}
    \caption{Sample trajectories of vitality process (transparent lines) and optimal consumption rate (solid lines) for relatively low bequest motive (left) and high bequest motive (right).}
    \label{fig:vitality_consumption}
\end{figure}

We emphasize that by incorporating the vitality mortality model into the optimal investment and consumption problem, new insights can be observed. 
Of course, our stylized model is overly simplified and can be extended in numerous ways. First, although an individual is able to track their health status, it is impossible to measure their vitality level accurately. To better reflect reality, noises should be introduced for the vitality level observed and lead to an optimal investment and consumption problem with partial information (for references on partial information, see for example, \cite{lakner1995utility, pham2001optimal, bauerle2007portfolio, ceci2012utility}, etc.).

Second, ambiguous beliefs regarding survival rates can influence individual decisions. In the model described above, individuals may face uncertainty in realizing their initial vitality levels $V(0)$ and the depletion rates $\delta$. For those averse to uncertainty, life-cycle investment and consumption planning will differ significantly. The concept of model uncertainty in life-cycle planning was pioneered by \cite{maenhout2004robust} and has more recently been applied to ambiguous mortality risk by \cite{young2016lifetime, shen2019life} within the context of traditional mortality laws. 

Third, it is possible to model medical expenditure as a function of vitality when studying the optimal consumption problem. Similarly, it is possible to express the vitality depletion rate $\mu_x(t)$ as a decreasing function of the medical expenditure when studying the optimal investment problem on personal health, see for example \cite{forster1989optimal, hugonnier2013health, cheng2023reference}, etc. One can explore numerous other directions where the optimization problem can be extended, highlighting the significant potential of the proposed vitality model in the area of lifetime planning.

\subsection{Subjective beliefs}
Numerous empirical studies have documented that individuals tend to have strong biases in their subjective beliefs about survival probabilities. In particular, individuals often underestimate survival probabilities at younger ages but overestimate them at older ages, as indicated in \cite{ludwig2013parsimonious}, \cite{groneck2016life}, and \cite{kalwij2021accuracy}. This empirical observation can be easily explained and understood from the vitality perspective. Specifically, when a person underestimates their initial vitality and underestimates their depletion rate, they will eventually underestimate the vitality level (and hence survival probability) at younger ages and, as time goes on, overestimate the vitality level (and hence survival probability) at older ages.

In addition, empirical studies indicate that individuals update the probability associated with pessimistic and optimistic sentiments when confronted with health shocks. This phenomenon can also be explained under our vitality-based modeling framework by the overestimation of a vitality jump, which results in an underestimation of the survival probability after the shock. Similar to a recent study by \cite{apicella2024behavioral}, our framework does not need to make specific assumptions about the difference in belief bias between different ages or the difference in optimism between different ages. In contrast to \cite{apicella2024behavioral}, where the authors construct their model solely to explain the belief bias, the proposed framework is constructed to accurately reflect the vitality experience of each individual while retaining the flexibility to explain belief bias.

Another interesting observation that can be reflected by our vitality model is the distinction between subjective belief in health status and subjective belief in life expectancy. Specifically, if a person correctly self-assesses as having an \emph{average} health level among the whole population, they may still over- or under-estimate their life expectancy if they falsely believe that their life expectancy will also be the \emph{average} among the whole population. Using the specification from Section \ref{sec:vitality_static_traditional} with $b = 0.0001744$ and $c = 1.082$ as an illustrative example, the population average life expectancy at age 60 is 17 years. However, a 60-year-old individual who believes that they have an average initial vitality level (which is one since the initial vitality follows an exponential distribution with a mean of one) will expect to live for 20.4 more years as given by equation \eqref{eq:vitality_tau_gompertz}. This observation can be summarized in the following proposition.

\begin{prop}
    Given that $V(0)\sim F_0$ where $F_0$ is a distribution function with positive support, the deterministic vitality depletion rate is strictly positive $\mu(t)>0$, and there are no diffusion and jump components, then 
    \begin{align*}
        \mathbb{E}[\tau] \leq \tau (\mathbb{E}[V(0)]),
    \end{align*}
    where $\tau(\cdot)$ is the conditional death time defined in Section \ref{sec:applications}. The average life expectancy is smaller than the life expectancy of an individual with an average vitality level.
\end{prop}

The proof immediately follows from the fact that the inverse function of a strictly increasing function is strictly concave, followed by Jensen's inequality. For example, for the vitality counterparts of traditional mortality law displayed in Section \ref{sec:vitality_static_traditional}, $Y(t)$ is a strictly increasing function of time, then, the death time given the initial vitality value is a concave function. The proposition suggests that even if an individual can objectively determine their health level relative to the population, it does not guarantee an unbiased estimate of their survival probabilities. However, if the individual judges their vitality level based on quantiles rather than the average (e.g., the median of the population vitality level), they will be able to arrive at an unbiased estimate of their life expectancy.

\subsection{Disability modeling}
It is natural to incorporate disability risk into our vitality model, similar to how mortality risk is incorporated. As the vitality level depletes, the risk of becoming disabled should increase over time. By setting a threshold above zero, we can classify an individual as disabled once their vitality level falls below this threshold. This extension allows the vitality model to capture the probability of both mortality and disability, and provide new insights into the interplay of mortality and morbidity risks.

\subsubsection{Permanent disability model}
Consider that there exists a threshold $\psi>0$ such that when the vitality falls below $\psi$, the person becomes disabled. Consider the specification with an exponentially distributed initial vitality (C1) and a Gompertz law specified depletion rate (C2). The resulting vitality model will form a permanent disability model, since $Y(t)$ is a strictly increasing function of time. The probability that a healthy person with initial vitality above $\psi$ will remain healthy for $T$ years is
\begin{align*}
   \Pr\left(V(0) -Y(T)>\psi\big|V(0)>\psi \right)&= \frac{\int_0^{\infty} e^{-v} \cdot \mathbb{1}_{v - Y(T) > \psi} \diff v}{\int_0^{\infty } e^{-v }\cdot \mathbb{1}_{v>\psi}  \diff v }\\
   &= \frac{ \int_{Y(T)+\psi}^{\infty} e^{-v}    \diff v }{\int_\psi^{\infty } e^{-v }   \diff v} = \exp( -Y(T)  ),
\end{align*}
which is equivalent to the survival probability, and the disability probability is simply $1 - \exp(-Y(T))$. This result holds even if we assume that each individual has a different disability threshold, such that the threshold $\Psi$ has a density function $\pi(\psi)$. In this case, the probability of remaining healthy for $T$ years is
\begin{align*}
    \frac{\int_0^{\infty} e^{-v}\cdot \int_0^{\infty} \pi(\psi) \cdot \mathbb{1}_{v-\psi - Y(T)>0}\diff \psi \diff v}{ \int_0^{\infty} e^{-v} \cdot \int_0^{\infty} \pi(\psi) \cdot \mathbb{1}_{v-\psi>0} \diff \psi \diff v}
    &= \frac{\int_{0}^{\infty} \pi(\psi)  e^{-(\psi + Y(T))}\diff \psi }{  \int_0^{\infty}\pi(\psi) \cdot e^{-\psi}\diff \psi } = \exp( -Y(T)  ).
\end{align*}
The above setting of the disability threshold $\omega$ can be extended in other directions, such as considering an age- and time-dependent threshold. We leave this for future research.

\subsubsection{Disability model with recovery}
Vitality models that have only an initial component (C1) and a non-decreasing trend component (C2) cannot accommodate the possibility of recovery if the disability is determined by the relative position of the vitality level with respect to a disability threshold. However, when the diffusion component (C3) is included, the model becomes capable of describing disability with the potential for recovery. 

For illustrative purposes, here we consider the simplest model that is able to capture the recovery probability, specifically, we consider the alternative Gompertz specification with a constant depletion rate and a diffusion component (C3) modeled as a Brownian motion multiplied by a constant volatility, along with a constant disability threshold. The result is summarized in the following proposition. 

\begin{prop}
Consider the vitality model with an initial vitality level $V(0)\sim F_0$, a constant vitality depletion rate $\mu(t) = \delta$, and a diffusion component as a Brownian motion with volatility $\sigma$. Consider a constant disability threshold $\psi$, then the $T$-year disabled probability for a healthy person aged $x$ is
\begin{align*}
    \int_{\psi}^{\infty} \frac{1}{1-F_0(\psi)}\cdot \left(\int_0^{\frac{v}{\sigma}} \int_{\frac{v-\psi}{\sigma} }^{\frac{v}{\sigma}} f(m,w;T) \diff w \diff m \right) \diff F_0(v),
\end{align*}
and the recovery probability for a disabled person aged $x$ is
    \begin{align*}
\int_0^{\psi}\frac{1}{F_0(\psi)} \cdot \left(\int_0^{\frac{v}{\sigma}} \int_0^{\frac{v-\psi}{\sigma} } f(m,w; T) \diff w \diff m \right)\diff F_0(v),
    \end{align*}
    where  
    \begin{align*}
    f(m,w;T) = \begin{cases}
        \frac{2(2m-w)}{T\sqrt{2\pi T}} e^{\frac{\delta w}{\sigma} - \frac{\delta^2T}{2\sigma^2} - \frac{(2m-w)^2}{2T}}, & w\leq m, m\geq 0\\ \\
        0, & \text{elsewhere}.
    \end{cases} 
    \end{align*}
\end{prop}

The proof is given in Appendix \ref{app:vitality_disability}. The probabilities can be evaluated by numerically solving the triple-integral. Notice that the choice of the disability threshold here does not impact the death probability, suggesting an approach to seamlessly integrate disability modeling and mortality modeling, which contrasts with traditional multi-state modeling for disability. Of course, the model above might be overly simplified; further exploration in model extension and empirical examination with disability data would be necessary.

\subsection{Cause-of-death modeling}

In our vitality-based modeling framework (equation~\eqref{eq:Vt}), death can be triggered either by the natural depletion of vitality (represented by components C2 and C3) or by a sudden jump due to a fatal accident (represented by component C4). Therefore, it is of practical interest to distinguish the cause of death based on the components of $V(t)$. In this subsection, we demonstrate how the vitality model can provide probabilities of different causes of death.

For illustrative purposes, we focus on the vitality model with components C1, C2, and C4 only. Specifically, we assume that component C2 follows the Gompertz law, while component C4 is modeled as a compound Poisson process with jump intensity $\lambda > 0$ and i.i.d. jump sizes $\left\{Z_{i}\right\}_{i \in \mathbb{N}_{+}}$. The jump sizes have a common cumulative distribution function denoted as ${F}_Z(z) = 1-\mathrm{e}^{-\alpha z}$ and survival function $\bar{F}_Z(z) = 1-F_Z(z)$. The resulting vitality model is expressed as follows:
\begin{align*}
V(t) = V(0) - \frac{b c^x}{\ln c}\left(c^t-1\right) - \sum_{i=1}^{N(t)} Z_i, \quad t\ge 0.
\end{align*}

Let us denote $\tau_J = \inf\{ t\ge 0: V(t) < 0 \}$ and $\tau_Y = \inf \{t \geq 0: V(t)= 0\}$. In other words, $\tau_J \cdot \mathbb{1}(\tau_J < \tau_Y)$ represents the death time at which death is caused by an accident, while $\tau_Y\cdot \mathbb{1}(\tau_Y < \tau_J)$ represents the death time at which death is caused by natural decay. It is clear that $\tau= \tau_Y \wedge \tau_J$. The next theorem gives the expressions of the Laplace transforms and densities of $\tau_J$ and $\tau_Y$, given that the initial vitality level is $v$. For ease of presentation, we let $\mathrm{Pr}_v := \mathrm{Pr}(\cdot| V(0) = v)$ and $\mathbb{E}_v[\cdot]: = \mathbb{E}[\cdot| V(0) = v]$. 

\begin{theorem} \label{decompose tau}
	For $q\ge 0$, the Laplace transforms of the time of death due to an accident and the time of death due to natural decay in vitality level are given by the following expressions:
	    \begin{align} \label{cpp: LT jump}
		     \mathbb{E}_v\left[\mathrm{e}^{-q\tau_{J}}\mathbb{1}\left(\tau_{J}<\tau_{Y}\right)\right]  =&\int_{0}^{t_v^{*}}\lambda \mathrm{e}^{-(q+\lambda)t-\alpha b^*_v(t)}\cdot I_{0}\left(2\sqrt{\alpha\lambda t b^*_v(t)}\right)\mathrm{d}t,
	\end{align}
and
\begin{align}\label{cpp: LT creep}
	\mathbb{E}_{v}\left[\mathrm{e}^{-q\tau_{Y}}\mathbb{1}\left(\tau_{Y}<\tau_{J}\right)\right]
	=&\int_{0}^{t^{*}_v}\mathrm{e}^{-\left(q+\lambda\right)t-\alpha  b^*_v(t)}\cdot bc^{x+t}\sqrt{\frac{\alpha\lambda t}{b^*_v(t)}}I_{1}\left(2\sqrt{\alpha\lambda t b^*_v(t) }\right)\mathrm{d}t+\mathrm{e}^{-(q+\lambda)t^{*}_v},
\end{align}
where $t^{*}_v= \frac{\ln(bc^{x}+v\ln c)-\ln b}{\ln c}-x$ (which refers to the remaining lifetime without any accidents/jumps), $b^*_v(t) =v-bc^{x}\left(c^{t}-1\right)/\ln c$, and
$$I_k(x):=\sum_{n=0}^{\infty} \frac{(x / 2)^{k+2 n}}{n ! \Gamma(k+n+1)}$$
is the modified Bessel function of the first kind. Moreover, by the uniqueness of the Laplace transform, it follows that
	\begin{align*} 
		\mathrm{Pr}_v\left(\tau_{J}\in\mathrm{d}t,\tau_{J}<\tau_{Y}\right) &= \mathbb{1}(t<t^{*})\lambda\mathrm{e}^{-\alpha b^*_v(t)-\lambda t}I_{0}\left(2\sqrt{\lambda\alpha t b^*_v(t)}\right)\mathrm{d}t,
	\end{align*}
and
\begin{align*}
	 \mathrm{Pr}_{v}\left(\tau_{Y}\in\mathrm{d}t,\tau_{Y}<\tau_{J}\right)
	=&\mathbb{1}(t<t^{*})\mathrm{e}^{-\lambda t-\alpha b^*_v(t) }\cdot bc^{x+t}\sqrt{\frac{\alpha\lambda t}{b^*_v(t)}}I_{1}\left(2\sqrt{\alpha\lambda t b^*_v(t)}\right)\mathrm{d}t+\mathrm{e}^{-\lambda t}\delta_{t^{*}}(\mathrm{d}t),
\end{align*}
where $\delta_{t_v^{*}}(\cdot)$ is the Dirac mass at $t_v^*$.
\end{theorem}

We note that the proof of Theorem \ref{decompose tau} can be modified to accommodate more general assumptions on the vitality dynamics $V(t)$. For instance, the trend component $Y(t)$ could be any general non-decreasing continuous function, and the jump component $J(t)$ could be a renewal process with positive jumps. Moreover, instead of a single jump process for the jump component, multiple jump processes can be incorporated to represent various causes of death, such as car accidents, murder, suicide, cancer, etc. In this way, a more detailed cause-of-death analysis can be performed.

\section{A Numerical Illustration} \label{sec:estimation}

In this section, we provide an illustrative example on the numerical estimation procedure of the proposed framework. We consider the jump-diffusion vitality model discussed in Section \ref{sec:vitality_jump}. Recall that the vitality process $V(t)$ is described by equation \eqref{eq:Vt_JumpDiffusion} while the resulting survival probability $\Pr(\tau > T)$ is given in equation \eqref{eq:SurvProb_JumpDiffusion}.

\subsection{Simplifications}
The computational requirements for evaluating equation \eqref{eq:SurvProb_JumpDiffusion} are quite intensive. For ease of computation, we make the following assumptions to achieve an approximation:
\begin{itemize}
    \item The force of mortality is constant between integer ages.
    \item All of the jumps are fatal accidents (i.e., $Z_i = \infty$ for all $i$). 
\end{itemize}

Under the constant force of mortality assumption, $Y(t)$ can be written as a piece-wise linear function:
\begin{align*}
    Y(t) = \alpha_i (t-\lfloor t \rfloor) + \beta_i,
\end{align*}
for $i\leq t < i+1$ and $i \in \{0,1,2, \ldots\}$, where $\beta_i = \int_0^i \mu_x(s)\diff s$ and $\alpha_i= \int_i^{i+1}\mu_x(s) \diff s$.
Using the well-known formula for the non-crossing probability of a linear boundary $\alpha t+ \beta$ before time $s$; that is,
\begin{align*}
    \Pr\left(B(t)<\alpha t + \beta \text{ for all } t \le s\big|B(s) = x\right) = 1-\exp \left\{-\frac{2 \beta(\alpha s + \beta - x)}{s}\right\},
\end{align*}
together with the fatal jump assumption, the survival probability conditional on an initial vitality $v>0$ is given by
\begin{align*}
    {}_kp_x(v) = \mathbb{E}\left[ \left(\prod_{s=1}^k \mathbb{1}_{B(s)< \frac{v - Y(s)}{\sigma} } \left( 1- e^{- \frac{2(v - Y(s-1) - \sigma B(s-1))(v-Y(s) - \sigma B(s))  }{\sigma^2} }\right) \right) \exp\left( -\int_0^k \lambda(s)\diff s \right) \right],
\end{align*}
for $k \in \mathbb{N}^+$. It follows that the $T$-year survival probability becomes
\[
    \Pr(\tau > T) = \int_{0}^{\infty} {}_Tp_x(v) \diff F_0(v).
\]

\subsection{Model estimation}
Given the complexity of vitality mortality models, using solely the mortality data from a national population is inadequate for estimation purposes. Specifically, such dataset does not distinguish between deaths caused by natural aging and those resulting from accidents. This limitation makes it challenging to differentiate between accidental deaths, captured by the jump component, and natural deaths, represented by the initial and trend components. In terms of parameter estimation, the intensity rate $\lambda(t)$ of the Poisson process $N(t)$ cannot be easily estimated without information that explicitly categorizes accidental deaths.

To address this issue, we utilize mortality data that are specifically related to accidental causes of death. In particular, we use the WONDER dataset from the Centers for Disease Control and Prevention (CDC) to determine the intensity rate $\lambda(t)$. The remaining parameters, $\theta_{-\lambda} = \theta \backslash \{\lambda(t)\}_{t=0,1,\cdots}$, are then estimated using the maximum likelihood estimation method. This approach allows us to separate the jump component from the rest of the components. More specifically, we estimate $\theta_{-\lambda}$ using the U.S. male mortality data from the Human Mortality Database (HMD)\footnote{Human Mortality Database. Max Planck Institute for Demographic Research (Germany), University of California, Berkeley (USA), and French Institute for Demographic Studies (France). Available at www.mortality.org.}. We focus on the cohort born in 1910, with age 60 and above during the period of 1970 to 2016. We  estimate $\lambda(t)$ using the average accidental death rates for age groups in ten-year intervals during the period of 1968 to 2016 from the CDC WONDER dataset.

The likelihood function of our maximum likelihood estimation is given as follows. Consider a cohort of individuals all aged $x$ at time $t=0$. Let $E_x$ be the initial exposure at time $0$ and $D_{x+t}$ be the observed death counts at age $x+t$ for $t \in \mathbb{N}^+$. We assume that $D_{x+t}$ follows a multinomial distribution, which leads to the log-likelihood function as 
\begin{align*}
    l(\theta) = \sum_{t \geq 0} D_{x+t} \cdot \ln \left[ \Pr(\tau > t) -  \Pr(\tau > t+1)  \right] + \ln E_x! - \sum_{t\geq 0} \ln D_{x+t}!,
\end{align*}
where $\theta$ is the set of all model parameters. Note that our estimation procedure differs from \cite{shimizu2021does} and \cite{shimizu2023survival}, who used the least squares method to estimate the parameters of their survival energy models.

\subsection{Numerical results}
Using the Gompertz law to specify the trend component of the vitality dynamics, the jump-diffusion vitality model is specified as
\[
    V(t) = V(0) - \int_{0}^{t} b c^{x+s} \diff s - \sigma B(t) - \sum_{i=1}^{N(t)} Z_i,
\]
where $V(0) \sim \text{Exp}(1)$. Given that the intensity rate of $N(t)$ has been externally determined and we assumed $Z_i = \infty$ for all $i$, the remaining parameters to be estimated are $b$, $c$, and $\sigma$. In addition to the vitality model, we also estimate the parameters of the Gompertz law for comparison purposes. 

Table \ref{tab:energy_parameters} reports the estimated values of $b$, $c$, and $\sigma$ from the vitality model and the Gompertz law for both males (left panel) and females (right panel). The maximized log-likelihood of each model is reported at the last row of the corresponding table. We observe that the estimated values of $b$ and $c$ are similar between the vitality model and the Gompertz law, indicating that both models accurately describe the underlying mortality curve via their assumed functional form (i.e., $\mu_x(t) = b c^{x+t}$). Compared to the Gompertz law, the vitality model additionally has an estimated value of $\sigma$ at 0.0098 and 0.0034 for males and females, respectively. These values reflect the random fluctuations in vitality over the lifetime of the corresponding cohort of males and females. Overall, the vitality model achieved a higher maximized log-likelihood than the Gompertz law for both genders.

\begin{table}[H]
    \centering
    \begin{tabular}{l|l l}
         \hline \hline   
         Model &  Vitality & Gompertz\\
         \hline
         $b$ & $1.4819 \times 10^{-4}$  & $1.7440 \times 10^{-4}$ \\
         $c$  & 1.0838 & 1.0820\\
         $\sigma$ & $9.8267 \times 10^{-3}$ & \\
         \hline   
         $l(\theta)$ & -2,057 & -2,591 \\
         \hline \hline
    \end{tabular}\qquad
    \begin{tabular}{l|ll}
    \hline \hline
    Model & Vitality & Gompertz\\
    \hline 
         $b$ & $1.8107 \times 10^{-5}$  & $1.9915 \times 10^{-5}$ \\
         $c$  & 1.1058 & 1.1048\\
         $\sigma$ & $3.4139  \times 10^{-3}$ & \\
         \hline   
         $l(\theta)$ & -3,375 & -3,901 \\
         \hline \hline
    \end{tabular}
    \caption{Parameter estimates and maximized log-likelihood of the vitality model and the Gompertz law for males (left) and females (right).}
    \label{tab:energy_parameters}
\end{table}

\section{Conclusion} \label{sec:Conclusion}

In this paper, we introduced a vitality-based modeling framework that offers a new perspective on mortality modeling. The proposed framework integrates four unique components to capture the complexities of vitality dynamics over an individual's life course. The adaptability of the four components allows for the incorporation of various factors affecting mortality, providing deeper insights into population mortality trends and individual health trajectories.

Through both static and dynamic parameter specifications of the proposed framework, we demonstrated its versatility in addressing various mortality modeling scenarios. Our framework not only replicates traditional mortality laws but also extends to more advanced stochastic mortality models, offering a unified structure for understanding mortality dynamics. An illustrative example is provided to showcase how a vitality model can be fitted using the maximum likelihood estimation method.

Our investigation suggests that the vitality-based modeling framework offers significant practical applications, such as life insurance valuation, disability modeling, cause-of-death modeling, explaining subjective beliefs on mortality, and deriving optimal consumption solutions. We emphasize that the proposed framework can be further utilized in more nuanced scenarios to improve results and decision-making in various fields.

While we used four components in the proposed framework, the specific structures considered for each component are not exhaustive in this paper. There are numerous possible extensions for articulating the four components. For instance, one could consider a multivariate stochastic process to govern the diffusion component, reflecting the fact that multiple underlying forces may drive the randomness in an individual's vitality. The jump component can also be expanded to consider jumps categorized by different causes of death.

We have demonstrated in this paper that the proposed framework is capable of replicating mortality laws and factor-based models. Further exploration can be conducted to show that intensity-based models, such as the Vasicek process applied to the force of mortality, can also be mimicked by the proposed framework. It would be interesting to understand how the first passage time problem connects with an continuous-time intensity-based model. This connection could potentially provide new insights and more accurate survival probability estimates.

There are several other potential applications that we did not consider in this paper. For instance, a vitality model with dynamic parameters can be used to manage longevity risk by developing hedging strategies and pricing longevity-linked securities. Additionally, since our framework models individual vitality, it can be employed to measure the basis risk between an individual insurance product and an index-based security. These extensions would provide further advancements to capital market solutions for longevity risk management.

Lastly, we acknowledge that the vitality dynamics of an individual cannot be directly observed and easily estimated. To improve the estimation of a vitality model, one would need to use personal health data or clinical data on vital measures. Potential sources of such data could include personal medical devices, such as smartwatches or mobile health trackers. We remark that this research direction is highly related to InsurTech and could inspire future studies on personalized insurance products.

\section*{Acknowledgement}

Xiaobai Zhu acknowledges the support from the Research Grants Council of the Hong Kong Special Administrative Region, China (CUHK 24615523), and National Natural Science Foundation of China (No. 12301613). Zijia Wang acknowledges the support from internal grants from The Chinese University of Hong Kong and National Natural Science Foundation of China (No. 12401623).



\bibliographystyle{apalike}
\bibliography{energy_bib}

\begin{thebibliography}{}

\bibitem[Aalen and Gjessing, 2004]{aalen2004survival}
Aalen, O.~O. and Gjessing, H.~K. (2004).
\newblock Survival models based on the {O}rnstein-{U}hlenbeck process.
\newblock {\em Lifetime data analysis}, 10(4):407--423.

\bibitem[Anderson, 1992]{anderson1992vitality}
Anderson, J.~J. (1992).
\newblock A vitality-based stochastic model for organism survival.
\newblock In DeAngelis, D. and Gross, L., editors, {\em Individual-based models
  and approaches in ecology: Populations, communities and ecosystems}, pages
  256--277. Chapman \& Hall, New York.

\bibitem[Anderson, 2000]{anderson2000vitality}
Anderson, J.~J. (2000).
\newblock A vitality-based model relating stressors and environmental
  properties to organism survival.
\newblock {\em Ecological monographs}, 70(3):445--470.

\bibitem[Anderson et~al., 2017]{anderson2017insights}
Anderson, J.~J., Li, T., and Sharrow, D.~J. (2017).
\newblock Insights into mortality patterns and causes of death through a
  process point of view model.
\newblock {\em Biogerontology}, 18:149--170.

\bibitem[Apicella and De~Giorgi, 2024]{apicella2024behavioral}
Apicella, G. and De~Giorgi, E.~G. (2024).
\newblock A behavioral gap in survival beliefs.
\newblock {\em Journal of Risk and Insurance}, 91(1):213--247.

\bibitem[Asmussen and Albrecher, 2010]{asmussen2010ruin}
Asmussen, S. and Albrecher, H. (2010).
\newblock {\em Ruin probabilities}, volume~14.
\newblock World scientific.

\bibitem[Asmussen et~al., 2016]{asmussen2016laplace}
Asmussen, S., Jensen, J.~L., and Rojas-Nandayapa, L. (2016).
\newblock On the laplace transform of the lognormal distribution.
\newblock {\em Methodology and Computing in Applied Probability}, 18:441--458.

\bibitem[Bauer et~al., 2010]{bauer2010pricing}
Bauer, D., B{\"o}rger, M., and Ru{\ss}, J. (2010).
\newblock On the pricing of longevity-linked securities.
\newblock {\em Insurance: Mathematics and Economics}, 46(1):139--149.

\bibitem[B{\"a}uerle and Rieder, 2007]{bauerle2007portfolio}
B{\"a}uerle, N. and Rieder, U. (2007).
\newblock Portfolio optimization with jumps and unobservable intensity process.
\newblock {\em Mathematical Finance}, 17(2):205--224.

\bibitem[B{\'e}gin et~al., 2024]{begin2023new}
B{\'e}gin, J.-F., Kapoor, N., and Sanders, B. (2024).
\newblock A new approximation of annuity prices for age--period--cohort models.
\newblock {\em European Actuarial Journal}, 14(2):697--703.

\bibitem[Biffis, 2005]{biffis2005affine}
Biffis, E. (2005).
\newblock Affine processes for dynamic mortality and actuarial valuations.
\newblock {\em Insurance: mathematics and economics}, 37(3):443--468.

\bibitem[Blackburn and Sherris, 2013]{blackburn2013consistent}
Blackburn, C. and Sherris, M. (2013).
\newblock Consistent dynamic affine mortality models for longevity risk
  applications.
\newblock {\em Insurance: Mathematics and Economics}, 53(1):64--73.

\bibitem[Cai and Kou, 2011]{cai2011option}
Cai, N. and Kou, S.~G. (2011).
\newblock Option pricing under a mixed-exponential jump diffusion model.
\newblock {\em Management Science}, 57(11):2067--2081.

\bibitem[Cairns et~al., 2006]{cairns2006two}
Cairns, A.~J., Blake, D., and Dowd, K. (2006).
\newblock A two-factor model for stochastic mortality with parameter
  uncertainty: theory and calibration.
\newblock {\em Journal of Risk and Insurance}, 73(4):687--718.

\bibitem[Cairns et~al., 2009]{cairns2009quantitative}
Cairns, A.~J., Blake, D., Dowd, K., Coughlan, G.~D., Epstein, D., Ong, A., and
  Balevich, I. (2009).
\newblock A quantitative comparison of stochastic mortality models using data
  from {E}ngland and {W}ales and the united states.
\newblock {\em North American Actuarial Journal}, 13(1):1--35.

\bibitem[Camarda, 2019]{camarda2019smooth}
Camarda, C.~G. (2019).
\newblock Smooth constrained mortality forecasting.
\newblock {\em Demographic Research}, 41:1091--1130.

\bibitem[Ceci, 2012]{ceci2012utility}
Ceci, C. (2012).
\newblock Utility maximization with intermediate consumption under restricted
  information for jump market models.
\newblock {\em International Journal of Theoretical and Applied Finance},
  15(06):1250040.

\bibitem[Chen et~al., 2022]{chen2022optimal}
Chen, C.-C., Chang, C.-C., Sun, E.~W., and Yu, M.-T. (2022).
\newblock Optimal decision of dynamic wealth allocation with life insurance for
  mitigating health risk under market incompleteness.
\newblock {\em European Journal of Operational Research}, 300(2):727--742.

\bibitem[Cheng et~al., 2023]{cheng2023reference}
Cheng, C., Hilpert, C., Szimayer, A., and Zweifel, P. (2023).
\newblock Reference health and investment decisions.
\newblock {\em Available at SSRN 4523775}.

\bibitem[Currie et~al., 2004]{currie2004smoothing}
Currie, I.~D., Durban, M., and Eilers, P.~H. (2004).
\newblock Smoothing and forecasting mortality rates.
\newblock {\em Statistical modelling}, 4(4):279--298.

\bibitem[Davydov and Linetsky, 2001]{davydov2001pricing}
Davydov, D. and Linetsky, V. (2001).
\newblock Pricing and hedging path-dependent options under the {CEV} process.
\newblock {\em Management science}, 47(7):949--965.

\bibitem[Di~Palo, 2023]{di2023closed}
Di~Palo, C. (2023).
\newblock On a closed-form expression and its approximation to {G}ompertz life
  disparity.
\newblock {\em Demographic Research}, 49:1--12.

\bibitem[Dodd et~al., 2018]{dodd2018smoothing}
Dodd, E., Forster, J.~J., Bijak, J., and Smith, P.~W. (2018).
\newblock Smoothing mortality data: The {E}nglish life tables, 2010--2012.
\newblock {\em Journal of the Royal Statistical Society: Series A (Statistics
  in Society)}, 181(3):717--735.

\bibitem[Dodd et~al., 2021]{dodd2021stochastic}
Dodd, E., Forster, J.~J., Bijak, J., and Smith, P.~W. (2021).
\newblock Stochastic modelling and projection of mortality improvements using a
  hybrid parametric/semi-parametric age--period--cohort model.
\newblock {\em Scandinavian Actuarial Journal}, 2021(2):134--155.

\bibitem[Durbin and Williams, 1992]{durbin1992first}
Durbin, J. and Williams, D. (1992).
\newblock The first-passage density of the {B}rownian motion process to a
  curved boundary.
\newblock {\em Journal of Applied Probability}, 29(2):291--304.

\bibitem[Finkelstein, 2012]{finkelstein2012discussing}
Finkelstein, M. (2012).
\newblock Discussing the strehler-mildvan model of mortality.
\newblock {\em Demographic Research}, 26:191--206.

\bibitem[Forster, 1989]{forster1989optimal}
Forster, B.~A. (1989).
\newblock Optimal health investment strategies.
\newblock {\em Bulletin of Economic Research}, 41(1).

\bibitem[Gavrilov and Gavrilova, 2019]{gavrilov2019new}
Gavrilov, L.~A. and Gavrilova, N.~S. (2019).
\newblock New trend in old-age mortality: {G}ompertzialization of mortality
  trajectory.
\newblock {\em Gerontology}, 65(5):451--457.

\bibitem[Gompertz, 1825]{gompertz1825}
Gompertz, B. (1825).
\newblock On the nature of the function expressive of the law of human
  mortality and on a new mode of determining the value of life contingencies.
\newblock {\em Philosophical transactions of the Royal Society of London},
  115:513--583.

\bibitem[Groneck et~al., 2016]{groneck2016life}
Groneck, M., Ludwig, A., and Zimper, A. (2016).
\newblock A life-cycle model with ambiguous survival beliefs.
\newblock {\em Journal of Economic Theory}, 162:137--180.

\bibitem[Gu{\'e}rin and Renaud, 2016]{guerin2016joint}
Gu{\'e}rin, H. and Renaud, J.-F. (2016).
\newblock Joint distribution of a spectrally negative {L}{\'e}vy process and
  its occupation time, with step option pricing in view.
\newblock {\em Advances in Applied Probability}, 48(1):274--297.

\bibitem[Haberman and Renshaw, 2012]{haberman2012parametric}
Haberman, S. and Renshaw, A. (2012).
\newblock Parametric mortality improvement rate modelling and projecting.
\newblock {\em Insurance: Mathematics and Economics}, 50(3):309--333.

\bibitem[Han et~al., 2024]{han2024production}
Han, J., Li, X., Sethi, S.~P., Siu, C.~C., and Yam, S. C.~P. (2024).
\newblock Production management with general demands and lost sales.
\newblock {\em Operations Research}.

\bibitem[Huang et~al., 2022]{huang2022modelling}
Huang, Z., Sherris, M., Villegas, A.~M., and Ziveyi, J. (2022).
\newblock Modelling {USA} age-cohort mortality: A comparison of multi-factor
  affine mortality models.
\newblock {\em Risks}, 10(9):183.

\bibitem[Hugonnier et~al., 2013]{hugonnier2013health}
Hugonnier, J., Pelgrin, F., and St-Amour, P. (2013).
\newblock Health and (other) asset holdings.
\newblock {\em Review of Economic Studies}, 80(2):663--710.

\bibitem[Hunt and Villegas, 2023]{hunt2023mortality}
Hunt, A. and Villegas, A.~M. (2023).
\newblock Mortality improvement rates: Modeling, parameter uncertainty, and
  robustness.
\newblock {\em North American Actuarial Journal}, 27(1):47--73.

\bibitem[Jennen, 1985]{jennen1985second}
Jennen, C. (1985).
\newblock Second-order approximations to the density, mean and variance of
  {B}rownian first-exit times.
\newblock {\em The Annals of Probability}, pages 126--144.

\bibitem[Jennen and Lerche, 1981]{jennen1981first}
Jennen, C. and Lerche, H.~R. (1981).
\newblock First exit densities of {B}rownian motion through one-sided moving
  boundaries.
\newblock {\em Zeitschrift f{\"u}r Wahrscheinlichkeitstheorie und verwandte
  Gebiete}, 55(2):133--148.

\bibitem[Jin and Wang, 2017]{jin2017first}
Jin, Z. and Wang, L. (2017).
\newblock First passage time for {B}rownian motion and piecewise linear
  boundaries.
\newblock {\em Methodology and Computing in Applied Probability}, 19:237--253.

\bibitem[Kalwij and Koc, 2021]{kalwij2021accuracy}
Kalwij, A. and Koc, V.~K. (2021).
\newblock Is the accuracy of individuals’ survival beliefs associated with
  their knowledge of population life expectancy?
\newblock {\em Demographic Research}, 45:453--468.

\bibitem[Kleinow, 2015]{kleinow2015common}
Kleinow, T. (2015).
\newblock A common age effect model for the mortality of multiple populations.
\newblock {\em Insurance: Mathematics and Economics}, 63:147--152.

\bibitem[Kyprianou, 2014]{kyprianou2014fluctuations}
Kyprianou, A.~E. (2014).
\newblock {\em Fluctuations of L{\'e}vy processes with applications:
  Introductory Lectures}.
\newblock Springer Science \& Business Media.

\bibitem[Lakner, 1995]{lakner1995utility}
Lakner, P. (1995).
\newblock Utility maximization with partial information.
\newblock {\em Stochastic processes and their applications}, 56(2):247--273.

\bibitem[Lee and Carter, 1992]{lee1992modeling}
Lee, R.~D. and Carter, L.~R. (1992).
\newblock Modeling and forecasting us mortality.
\newblock {\em Journal of the American statistical association},
  87(419):659--671.

\bibitem[Li et~al., 2021a]{li2021gompertz}
Li, H., Tan, K.~S., Tuljapurkar, S., and Zhu, W. (2021a).
\newblock Gompertz law revisited: Forecasting mortality with a multi-factor
  exponential model.
\newblock {\em Insurance: Mathematics and Economics}, 99:268--281.

\bibitem[Li et~al., 2021b]{li2021constructing}
Li, J. S.-H., Li, J., Balasooriya, U., and Zhou, K.~Q. (2021b).
\newblock Constructing out-of-the-money longevity hedges using parametric
  mortality indexes.
\newblock {\em North American Actuarial Journal}, 25(sup1):S341--S372.

\bibitem[Li and Anderson, 2009]{li2009vitality}
Li, T. and Anderson, J.~J. (2009).
\newblock The vitality model: A way to understand population survival and
  demographic heterogeneity.
\newblock {\em Theoretical Population Biology}, 76(2):118--131.

\bibitem[Lin and Willmot, 2000]{lin2000moments}
Lin, X.~S. and Willmot, G.~E. (2000).
\newblock The moments of the time of ruin, the surplus before ruin, and the
  deficit at ruin.
\newblock {\em Insurance: Mathematics and Economics}, 27(1):19--44.

\bibitem[Lo, 2013]{lo2013wkb}
Lo, C.-F. (2013).
\newblock {WKB} approximation for the sum of two correlated lognormal random
  variables.
\newblock {\em Applied Mathematical Sciences}, 7(128):6355--6367.

\bibitem[Ludwig and Zimper, 2013]{ludwig2013parsimonious}
Ludwig, A. and Zimper, A. (2013).
\newblock A parsimonious model of subjective life expectancy.
\newblock {\em Theory and Decision}, 75:519--541.

\bibitem[Maenhout, 2004]{maenhout2004robust}
Maenhout, P.~J. (2004).
\newblock Robust portfolio rules and asset pricing.
\newblock {\em Review of financial studies}, 17(4):951--983.

\bibitem[Makeham, 1860]{makeham1860law}
Makeham, W.~M. (1860).
\newblock On the law of mortality and the construction of annuity tables.
\newblock {\em Journal of the Institute of Actuaries}, 8(6):301--310.

\bibitem[Merton, 1969]{merton1969lifetime}
Merton, R.~C. (1969).
\newblock Lifetime portfolio selection under uncertainty: The continuous-time
  case.
\newblock {\em The review of Economics and Statistics}, pages 247--257.

\bibitem[Milevsky and Promislow, 2001]{milevsky2001mortality}
Milevsky, M.~A. and Promislow, S.~D. (2001).
\newblock Mortality derivatives and the option to annuitise.
\newblock {\em Insurance: Mathematics and Economics}, 29(3):299--318.

\bibitem[Missov, 2013]{missov2013gamma}
Missov, T.~I. (2013).
\newblock Gamma-{G}ompertz life expectancy at birth.
\newblock {\em Demographic Research}, 28:259--270.

\bibitem[Pham and Quenez, 2001]{pham2001optimal}
Pham, H. and Quenez, M.-C. (2001).
\newblock Optimal portfolio in partially observed stochastic volatility models.
\newblock {\em Annals of Applied Probability}, pages 210--238.

\bibitem[Salminen, 1988]{salminen1988first}
Salminen, P. (1988).
\newblock On the first hitting time and the last exit time for a {B}rownian
  motion to/from a moving boundary.
\newblock {\em Advances in Applied Probability}, 20(2):411--426.

\bibitem[Sharrow and Anderson, 2016]{sharrow2016quantifying}
Sharrow, D.~J. and Anderson, J.~J. (2016).
\newblock Quantifying intrinsic and extrinsic contributions to human longevity:
  Application of a two-process vitality model to the human mortality database.
\newblock {\em Demography}, 53(6):2105--2119.

\bibitem[Shen and Su, 2019]{shen2019life}
Shen, Y. and Su, J. (2019).
\newblock Life-cycle planning with ambiguous economics and mortality risks.
\newblock {\em North American Actuarial Journal}, 23(4):598--625.

\bibitem[Shimizu et~al., 2021]{shimizu2021does}
Shimizu, Y., Minami, Y., and Ito, R. (2021).
\newblock Why does a human die? {A} structural approach to cohort-wise
  mortality prediction under survival energy hypothesis.
\newblock {\em ASTIN Bulletin: The Journal of the IAA}, 51(1):191--219.

\bibitem[Shimizu et~al., 2023]{shimizu2023survival}
Shimizu, Y., Shirai, K., Kojima, Y., Mitsuda, D., and Inoue, M. (2023).
\newblock Survival energy models for mortality prediction and future prospects.
\newblock {\em ASTIN Bulletin: The Journal of the IAA}, 53(2):377--391.

\bibitem[Shreve et~al., 2004]{shreve2004stochastic}
Shreve, S.~E. et~al. (2004).
\newblock {\em Stochastic calculus for finance II: Continuous-time models},
  volume~11.
\newblock Springer.

\bibitem[Siegmund, 1986]{siegmund1986boundary}
Siegmund, D. (1986).
\newblock Boundary crossing probabilities and statistical applications.
\newblock {\em The Annals of Statistics}, pages 361--404.

\bibitem[Skiadas and Skiadas, 2010]{skiadas2010comparing}
Skiadas, C.~H. and Skiadas, C. (2010).
\newblock Comparing the {G}ompertz-type models with a first passage time
  density model.
\newblock {\em Advances in Data Analysis: Theory and Applications to
  Reliability and Inference, Data Mining, Bioinformatics, Lifetime Data, and
  Neural Networks}, pages 203--209.

\bibitem[Strehler and Mildvan, 1960]{strehler1960general}
Strehler, B.~L. and Mildvan, A.~S. (1960).
\newblock General theory of mortality and aging: A stochastic model relates
  observations on aging, physiologic decline, mortality, and radiation.
\newblock {\em Science}, 132(3418):14--21.

\bibitem[Ungolo et~al., 2024]{ungolo2023estimation}
Ungolo, F., Garces, L. P. D.~M., Sherris, M., and Zhou, Y. (2024).
\newblock Estimation, comparison, and projection of multifactor age--cohort
  affine mortality models.
\newblock {\em North American Actuarial Journal}, 28(3):570--592.

\bibitem[Vaupel et~al., 1979]{vaupel1979impact}
Vaupel, J.~W., Manton, K.~G., and Stallard, E. (1979).
\newblock The impact of heterogeneity in individual frailty on the dynamics of
  mortality.
\newblock {\em Demography}, 16(3):439--454.

\bibitem[Wong et~al., 2017]{wong2017managing}
Wong, T.~W., Chiu, M.~C., and Wong, H.~Y. (2017).
\newblock Managing mortality risk with longevity bonds when mortality rates are
  cointegrated.
\newblock {\em Journal of Risk and Insurance}, 84(3):987--1023.

\bibitem[Yamazaki, 2017]{yamazaki2017inventory}
Yamazaki, K. (2017).
\newblock Inventory control for spectrally positive {L}{\'e}vy demand
  processes.
\newblock {\em Mathematics of Operations Research}, 42(1):212--237.

\bibitem[Young and Zhang, 2016]{young2016lifetime}
Young, V.~R. and Zhang, Y. (2016).
\newblock Lifetime ruin under ambiguous hazard rate.
\newblock {\em Insurance: Mathematics and Economics}, 70:125--134.

\bibitem[Zhou et~al., 2022]{zhou2022stochastic}
Zhou, H., Zhou, K.~Q., and Li, X. (2022).
\newblock Stochastic mortality dynamics driven by mixed fractional {B}rownian
  motion.
\newblock {\em Insurance: Mathematics and Economics}, 106:218--238.

\bibitem[Zhu and Zhou, 2023]{zhu2023smooth}
Zhu, X. and Zhou, K.~Q. (2023).
\newblock Smooth projection of mortality improvement rates: A bayesian
  two-dimensional spline approach.
\newblock {\em European Actuarial Journal}, 13(1):277--305.

\end{thebibliography}


\newpage

\appendix
\section{Technical Details} \label{app:Proofs} 


\subsection{Plateau Death Models} \label{app:ProofsPlateau}

\subsubsection*{\underline{\textbf{Gamma-Gompertz Model}}} \label{app:gammagompertz}

Consider a hazard rate model with the hazard rate function being $h(t) = Z\cdot \mu_x(t)$, where $Z$ follows a Gamma distribution with shape parameter $\alpha$ and \emph{rate} parameter $\phi$. Note that the conditional $T$-year survival probability given $Z=z$ can be expressed as $\exp\left(-\int_0^Tz\cdot \mu_x(t)\diff t \right)$, then the unconditional $T$-year survival probability is

\begin{align*}
  \Pr(\tau > T) =&\int_0^{\infty} \left(e^{ -\int_0^T z\cdot \mu_x(t)\diff t  }\right)\cdot  \frac{\phi^{\alpha}  \cdot  z^{\alpha-1}\cdot  e^{-z\cdot \phi}  }{\Gamma(\alpha) }\diff z\\
    =& \int_0^{\infty} \underset{\text{density function of a Gamma distribution}}{\underbrace{ \frac{ e^{ -\left( \int_0^T\mu_x(t)\diff t + \phi \right)\cdot z } \cdot z^{\alpha-1}}{\Gamma(\alpha) } \cdot \left( \int_0^T\mu_x(t)\diff t + \phi \right)^{\alpha} }}\cdot  \frac{\phi^{\alpha}  }{\left( \int_0^T\mu_x(t)\diff t + \phi \right)^{\alpha}} \diff z\\
   =& \left(1 + \frac{ \int_0^T\mu_x(t)\diff t  }{\phi}   \right)^{-\alpha}.
\end{align*}
In particular, if $\mu_x(t) = bc^{x+t}$, it is known as the Gamma-Gompertz law.

\subsubsection*{\underline{\textbf{Exponential Initial Vitality and Gamma Decay Rate}}}

If we assume that the initial vitality level $V(0)\sim \text{Exp}(1)$ and $Y(t)$ follows
\begin{align*}
    \diff Y(t) = Z\cdot \mu_x(t)\diff t, \quad t\ge 0,
\end{align*}
where $Z$ follows a Gamma distribution with shape $\alpha$ and rate $\phi$ and independent of $V(0)$, then we have
\begin{align*}
    \Pr(\tau > T) &= \Pr\left(V(0) - Z\cdot \int_0^T\mu_x(s)\diff s > 0 \right)\\
    &= \int_0^{\infty} \Pr\left( V(0) > z\cdot \int_0^T\mu_x(s)\diff s \right) \cdot  \frac{\phi^{\alpha}  \cdot  z^{\alpha-1}\cdot  e^{-z\cdot \phi}  }{\Gamma(\alpha) } \diff z\\
    &= \int_0^{\infty} \exp\left(  - z\cdot \int_0^T\mu_x(s)\diff s \right) \cdot  \frac{\phi^{\alpha}  \cdot  z^{\alpha-1}\cdot  e^{-z\cdot \phi}  }{\Gamma(\alpha) }\diff z .
\end{align*}
The remaining steps follow exactly the same as before, and we will end up with the same survival function as in equation \eqref{eq:SurvProb_Plateau}.

\subsubsection*{\underline{\textbf{Dagum Decay Rate}}}

Assume that the initial vitality is a constant $v$, and
\begin{align*}
    \diff Y(t) = Z \cdot \mu_x(t)\diff t, \quad t\ge 0,
\end{align*}
where $Z$ follows a Dagum (Inverse-Burr) distribution with parameters $p$, $a$, $b$ and distribution function $F_Z(\cdot)$, then we have
\begin{align*}
    \Pr(\tau>T) &= \int_0^{\infty} \mathbb{1}\left( v- z\cdot \int_0^T\mu_x(s)\diff s  > 0  \right) \diff F_Z(z)\\
    &= \int_0^{\infty} \mathbb{1}\left( z< \frac{v}{\int_0^T\mu_x(s)\diff s } \right)\diff F_Z(z)\\
    &= \left( 1 + \left(\frac{ v }{b\cdot \left(\int_0^T\mu_x(s)\diff s\right) } \right)^{-a} \right)^{-p}.
\end{align*}
If $a = 1$, $p = \alpha$, and $b = \frac{v}{\phi}$, then we recover the same survival function as in equation \eqref{eq:SurvProb_Plateau}.

\subsection{Alternative Expression of Mortality Laws} \label{alternative mortality}

Noting that the vitality model given in equation \eqref{eq:Vt_AlterStochastic} is a special case of the class of SNLPs, one can adopt the fluctuation theory of SNLP (see \cite{kyprianou2014fluctuations} for more details) to study $\tau$. For example, suppose the distribution of $\{Z_i\}_{i\ge 1}$ is a mixture of exponential distributions with probability density function 
\begin{align*}
   f_Z(z)= \sum_{i=1}^{n} p_{i} \alpha_{i} \mathrm{e}^{-\alpha_{i} z}, \quad z>0,
\end{align*}
where $\alpha_i > 0$ and $\sum_{i=1}^n p_i =1$.~In this case, it is known from \cite{guerin2016joint} that the Laplace transform of $\tau$ given that $V(0) = v$ is 
\begin{align*}
\mathbb{E}_{v}\left[\mathrm{e}^{-q\tau}\right]&=q\sum_{i=1}^{n+2}\frac{\mathrm{e}^{\theta_{i}^{(q)}v}}{\theta_{i}^{(q)}g\left(\theta_{i}^{(q)}\right)}-\frac{q}{\theta_{1}^{(q)}}\sum_{i=1}^{n+2}\frac{\mathrm{e}^{\theta_{i}^{(q)}v}}{g\left(\theta_{i}^{(q)}\right)}, \quad q>0,
\end{align*}
where $\theta_{n+2}^{(q)}<\theta_{n+1}^{(q)}<\theta_{n}^{(q)}\ldots<\theta_{2}^{(q)}\leq0\leq\theta_{1}^{(q)}$ are the roots (in $y$) to the equation 
\begin{align*}
 \frac{1}{2}\sigma^{2}y^{2}-\delta y+\sum_{i=1}^{n}\frac{\lambda p_{i}\alpha_{i}}{y+\alpha_{i}}-\lambda -q =0,
\end{align*}
and 
\begin{align*}
    g(y)=\sigma^{2}y-\delta-\sum_{i=1}^{n}\frac{\lambda p_{i}\alpha_{i}}{\left(y+\alpha_{i}\right)^{2}}.
\end{align*}
Then, by either analytically or numerically invert the above Laplace transform with respect to $q$, one can obtain the distribution function of $\tau$.

\subsection{Vitality Models with Dynamic Parameters} \label{dynamic vitality}
When $Y(t)$ is non-decreasing, for example, the integral of correlated log-normally distributed random variables as in Section \ref{sec:vitality_dynamic_cohort}, then 
\begin{align*}
    \Pr(\tau > T) &= \Pr(V_y(0) > Y(T;y)) \\
    &= \int_0^{\infty} \Pr(V_y(0) > Y(T;y)|\gamma) \cdot f_y(\gamma) \diff \gamma \\
    &= \int_0^{\infty} \mathrm{Pr} \left(\mathrm{e}_\gamma  > Y(T;y)\right) \cdot f_y(\gamma) \diff \gamma \\
     &= \int_0^{\infty} \mathbb{E} \left[ \mathrm{e}^{-\gamma Y(T;y)} \right] \cdot f_y(\gamma) \diff \gamma \\
        &= \mathbb{E}\left[ \exp\left( - \Gamma(y) \cdot Y(T;y)\right) \right],
\end{align*}
where $\mathrm{e}_\gamma$ is an exponential random variable with parameter $\gamma$, and $f_y(\gamma)$ is the density function for $\Gamma(y)$ of cohort $y$.

\section{Monte Carlo Valuation} \label{app:MonteCarlo}
Here we employ the Monte Carlo valuation to the survival probability of the model in Section \ref{sec:vitality_jump}. Denote $H(t, v) = v - Y(t)$ for ease of notation, we can apply the following procedure to approximate ${}_tp_x(v)$:
\begin{enumerate}
	\item Select $n + 1$ time points $\{t_i\}_{i = 0, \cdots, n}$ such that $0 = t_0 < t_1 < ... < t_n = t$. 
	\item Given that $N(t) = k > 0$, generate the first jump tim $v_1$ with distribution function at time $r_1<t$ being $\frac{\int_{0}^{r_1}\lambda(s) \diff s }{ \int_0^t\lambda(s)\diff s  }$, then we can iteratively draw the inter-arrival time between the $(i-1)$-th and the $i$-th jumps with distribution function for $r_i$, $r_{i-1}<r_i<t$, being $\frac{ \int_{r_{i-1}}^{r_i} \lambda(s)\diff s }{ \int_{r_{i-1}}^{t} \lambda(s)\diff s }$. Let $s_1 = r_1$, $s_2 = r_1 + r_2$, $\cdots$, $s_k = \sum_{i=1}^k r_i$, which are the arrival time for each jump.
 
	\item By combining the two sets of partitions $\{\tilde{u}_i\}_{i=0,\cdots, n+k}  = \{t_i\}_{i=0,\cdots, n} \cup \{s_i\}_{i=1,\cdots, k}$ and re-ordering this set to have a new set denoted as $\{u_i\}_{i=0,\cdots, n+k}$, which is in ascending order such that $0 = u_0 < u_1 < ... < u_{n-1+k} < u_{n+k} = t$.
	
	\item Generate $\{Z_i\}_{i=1,\cdots, k}$ jump sizes from $k$ i.i.d. random variables with the cdf $F_Z(\cdot)$, denote the simulated samples as $z_1, z_2, \ldots z_k$.
	
	\item Compute the following quantities: 
	\begin{align*}
    	&b_0^+ := \frac{H(0, v)}{\sigma},\\
    	&b_1^+ := \frac{H(u_1, v)}{\sigma} - \sum_{j=1}^{k} \frac{z_j }{\sigma}\mathbb{1}_{\{s_j \le u_1\}}, \ \ldots, \ b_i^+ := \frac{H(u_i, v)}{\sigma} - \sum_{j=1}^{k} \frac{z_j }{\sigma} \mathbb{1}_{\{s_j \le u_i\}},\\
        &b_1^- := \frac{H(u_1,v)}{\sigma} - \sum_{j=1}^{k} \frac{z_j }{\sigma} \mathbb{1}_{\{s_j < u_1\}}, \ \ldots, \ b_i^- := \frac{H(u_i,v)}{\sigma} - \sum_{j=1}^{k} \frac{z_j }{\sigma} \mathbb{1}_{\{s_j < u_i\}},
	\end{align*}
	 for $i = 1, 2, ..., n+k$.
	\item Generate $(y_1, y_2, ... y_{n+k-1})$ from the $(n+k-1)$-dimensional multivariate normal distribution of $B(u_1), B(u_2), ..., B(u_{n+k-1})$ and let $y_0 = 0$.

\item Compute
\begin{align*}
&Q_1 = \prod_{i=1}^{n+k} \mathbb{1}\left(y_{i}<b_{i}^+\right)\left(1-\exp \left\{-\frac{2\left(b_{i-1}^{+}-y_{i-1}\right)\left(b_{i}^{-}-y_{i}\right)}{u_{i}-u_{i-1}}\right\}\right),
\end{align*}
where $\Phi$ is the cumulative distribution function of a standard normal distribution.

\item Repeat Steps 1 -- 7 for $m$ times, to obtain $Q_i$, $i=1, \cdots, m$ (from Step 7), and the survival probability (conditional on initial vitality of $v$) can be approximated by 
\begin{align*}
{}_tp_x(v) \approx \frac{1}{m} \sum_{i=1}^{m} Q_{i}.
\end{align*}

\end{enumerate}

{\color{black}
\section{Optimal Investment and Consumption}\label{app:optimal}
\subsection{Merton's Problem - Infinite Time}\label{app:optimal_merton_infinite}
The value function of Merton's problem under an infinite time horizon at time $t$ is
\begin{align*}
 J(a) =  \max_{\pi, \zeta \in \Pi} \mathbb{E}_{t,a} \left[\int_t^{\infty} e^{-\beta(s-t)}\cdot u(\zeta(s)) \diff s\right].
\end{align*}
The associated Hamilton–Jacobi–Bellman (HJB) equation can be expressed as
\begin{align*}
    \max_{\pi, \zeta} \left\{ J_a\cdot\left[a \cdot r + a\cdot \pi \cdot \theta\cdot \sigma_S - \zeta \right]  +  J_{aa} \cdot \frac{a^2\cdot \sigma_S^2 \cdot \pi^2}{2} + u(\zeta)\right\}  - \beta\cdot J = 0.
\end{align*}
The value function can be solved explicitly as
\begin{align*}
    J(a) = F\cdot \ln a + G, \quad \text{where } F = \frac{1}{\beta},  \text{ } G = \exp\left(\frac{[r + \theta^2/2- \beta ]/\beta + \ln \beta}{\beta}\right),
\end{align*}
and the optimal strategies are
\begin{align*}
    \pi^*(t) = \frac{\theta}{\sigma_S}, \quad \zeta^*(t) = \beta \cdot A(t).
\end{align*}

\subsection{Merton's Problem - Mortality}\label{app:optimal_merton_mortality}
Consider again Merton's problem with finite planning horizon $T$ and with mortality modeled by a deterministic force of mortality function $\mu_{x}$. The value function at time $t$ (for a person age $x$ at time $0$) is
\begin{align*}
    J(t,a) = \max_{\pi, \zeta\in \Pi} \mathbb{E}_{t,a} \left[\int_{t}^{T} e^{-\beta (s-t)}\cdot e^{-\int_{t}^s \mu_{x+y}\diff y} \cdot \left(u(\zeta(s))  + \mu_{x+s}\cdot \lambda_1 \cdot u(A(s)) \right)\diff s + \lambda_2 \cdot u(A(T)) \right],
\end{align*}
where $\lambda_1$ and $\lambda_2$ are the weighting parameters for the bequest and the end-of-period wealth, respectively. The associated HJB equation is
\begin{align*}
    J_t + \max_{\pi, \zeta} \left\{ J_a\cdot\left[a \cdot r + a\cdot \pi \cdot \theta\cdot \sigma_S - \zeta \right]  +  J_{aa} \cdot \frac{a^2\cdot \sigma_S^2 \cdot \pi^2}{2} + u(\zeta)  + \mu_{x+t}\cdot \lambda_1 \cdot u(a) \right\}  - (\beta + \mu_{x+t}) \cdot J = 0
\end{align*}
By conjecture that $J(t,a) = F(t) \cdot \ln a + G(t)$ with terminal condition
\begin{align*}
    &J(T, a) = \lambda_2 \cdot \ln a, \quad \implies F(T) = \lambda_2 \quad \text{and} \quad G(T) = 0,
\end{align*}
and by the first order condition, the optimal strategies are
\begin{align*}
    \pi^* &= \frac{-J_a\cdot \theta}{J_{aa}\cdot a \cdot \sigma_S} = \frac{\theta}{\sigma_S}, \quad \zeta^* = \frac{1}{J_a} = \frac{a}{F(t)}.
\end{align*}
Substitute into the HJB equation and simplify:
\begin{align*}
    &F'(t)\ln a + G'(t) + \frac{F(t)}{a} \cdot \left(a\cdot r + a\cdot \theta^2 - \frac{a}{F(t)}\right) - \frac{F(t)}{a^2} \cdot \frac{a^2 \theta^2}{2}\\
    & \qquad + \ln a - \ln F(t) + \mu_{x+t}\cdot \lambda_1\cdot u(a) - (\beta+ \mu_{x+t})(F(t)\ln a + G(t)) = 0.
\end{align*}
Grouping terms with $a$ and without $a$, we have
\begin{align*}
    &F'(t) + 1  + \mu_{x+t}\cdot \lambda_1 - (\beta+\mu_{x+t})\cdot F(t) = 0\\
    \implies & F(t) =  \lambda_2 \cdot e^{-\int_t^T \beta + \mu_{x+y}\diff y } + \int_t^T e^{-\int_t^s\beta + \mu_{x+y}\diff y }\cdot (1+ \mu_{x+s} \cdot \lambda_1 )\diff s   \\
    &G'(t) + F(t)\cdot r + F(t) \cdot \frac{\theta^2}{2} -1 - \ln F(t) - (\beta+\mu_{x+t})\cdot G(t) = 0\\
    \implies & G(t) = \int_t^T e^{-\int_t^s \beta + \mu_{x+y}\diff y}\cdot \left( F(s)\cdot \left(r + \frac{\theta^2}{2} \right) - 1 - \ln F(s)  \right)\diff s.
\end{align*}
In particular, when $T\rightarrow \infty$, $F(t)$ becomes
\begin{align*}
    F(t) &= \int_t^{\infty} e^{-\beta(t-s)}\cdot e^{-\int_t^s \mu_{x+y}\diff y} \diff s + \int_t^{\infty} e^{-\beta(t-s)}\cdot e^{-\int_t^s \mu_{x+y}\diff y} \cdot \mu_{x+s}\cdot \lambda_1 \diff s\\
    &= \bar{a}_{x+t} + \lambda_1\cdot \bar{A}_{x+t},
\end{align*}
where $\bar{a}_{x}$ and $\bar{A}_{x}$ are annuity price and insurance price issued to a person aged $x$ with \$1 paid continuously, evaluated using a discount rate of $\beta$.  

\subsection{Merton's Problem with Vitality}
\vspace{0.3cm}
\begin{definition}(Admissible Strategies) The pair of strategies $(\pi(\cdot), \zeta(\cdot))$ is called an admissible strategy, if it satisfies the following conditions:
\begin{enumerate}
    \item $\pi(\cdot)$ and $\zeta(\cdot)$ are progressively measurable;
    \item $\mathbb{E}_{0,a,v}\left[ \int_0^t \pi^2(s) + \zeta^2(s) \diff s \right]<\infty$ and $A(t)>0$ for any $t\in[0, \infty)$ and $v > 0$.
    \item associated with $(\pi(\cdot), \zeta(\cdot))$, the state equation $A(\cdot)$ has a unique strong solution.
\end{enumerate}
We denote the set of all admissible strategies by $\Pi$.\\
\end{definition}
The associated HJB equation for the value function $J(a,v)$ is
\begin{align*}
    \max_{\pi, \zeta} \left\{ J_a\cdot\left[a \cdot r + a\cdot \pi \cdot \theta\cdot \sigma_S - \zeta \right]  +  J_{aa} \cdot \frac{a^2\cdot \sigma_S^2 \cdot \pi^2}{2} - J_v\cdot \delta + J_{vv}\cdot \frac{\sigma_V^2}{2} + u(\zeta)\right\} - \beta \cdot J = 0
\end{align*}
with boundary conditions
\begin{align*}
    \lim_{v\rightarrow \infty} J(a,v) &= F\ln a + G\\
    \lim_{v\rightarrow 0} J(a,v) &= \lambda \ln a,
\end{align*}
where $F$ and $G$ are defined in the value function in Section \ref{app:optimal_merton_infinite}. The first line corresponds to the case where the individual will never die because the vitality approaches infinite, and thus the problem collapses to Merton's problem in Section \ref{app:optimal_merton_infinite}. The second line corresponds to when the individual is already dead since the vitality has been exhausted, and thus the individual's utility is reflected only through the bequest.\\

By the first-order condition, the optimal strategies can be expressed as:
 \begin{align*}
     \pi^* &=  \frac{ - J_a \cdot a \cdot \theta \cdot \sigma_S  }{J_{aa} \cdot a^2\cdot \sigma_S^2} \quad \text{and} \quad \zeta^* = \frac{1}{J_a}.
 \end{align*}
By conjecturing on the form of the value function $J(v,a) = f(v)\cdot \ln a + g(v)$, where $f(v)$ and $g(v)$ are functions of the vitality level, then 
\begin{align*}
    \pi^* &= \frac{ \theta  }{\sigma_S} \quad \text{and} \quad \zeta^* = \frac{a}{f(v)}.
\end{align*}
Substituting the expression back to the HJB equation and collecting the terms, we have
\begin{align*}
    &\delta \cdot f'(v) - f''(v) \cdot \frac{\sigma_V^2}{2} - 1 + \beta \cdot f(v) = 0\\
    &f(v)\cdot \left[r + \frac{\theta^2}{2}\right] - g'(v)\cdot \delta + g''(v)\cdot\frac{\sigma_V^2}{2} - \ln f(v) - \beta \cdot g(v) =0
\end{align*}
Therefore, we have
\begin{align*}
    f(v) = \left(\lambda - \frac{1}{\beta}\right) \cdot e^{k_1 v} + \frac{1}{\beta}, \quad  k_1 = \left\{\begin{array}{ll}
       \frac{ \delta - \sqrt{\delta^2 + 2 \sigma^2_V  \beta} }{ \sigma^2_V }<0,  & \sigma_V>0 \\
       \frac{-\beta}{\delta},  &  \sigma_V = 0.
    \end{array}\right.
\end{align*}
In particular, when $\sigma_V = 0$, the optimal consumption strategy becomes
\begin{align*}
    f(t) = \bar{a}_{\annuity{t^*}} + \lambda\cdot e^{-\beta t^*},
\end{align*}
where $t^* = \frac{v}{\delta}$ is the remaining life, and $\bar{a}_{\annuity{t^*}}$ is the annuity-certain evaluated with discount rate $\beta$. The expression of $f(t)$ has a similar interpretation with $F(t)$ in Section \ref{app:optimal_merton_mortality} but based on the remaining life $t^*$ ($e^{-\beta t^*}$ is the present value of \$1 pays at time $t^*$). 

}

\section{Disability with Recovery} \label{app:vitality_disability}
\begin{proof}
The $T$-year recovery probability for a disabled person aged $x$ can be expressed as

\begin{align}
     &\Pr\left( \underset{\text{healthy at $T$}}{\underbrace{\left\{V(T) > \psi\right\}}} \ \cap \underset{\text{remains alive until T} }{\underbrace{ \left\{ V(t)> 0, \ \forall t\in[0,T]\right\}}}\bigg| \underset{\text{disabled initially}}{\underbrace{V(0)<\psi }} \right) \nonumber \\ \nonumber \\
    =& \int_{0}^{\psi} \frac{\Pr\left( \left\{v - Y(T)-W(T)>\psi\right\} \text{ and } \left\{v-Y(t)-W(t) > 0   \text{ for any } t\in [0,T] \right\} \right)}{F_0(\psi) }\diff F_0(v). \label{eq:vitality_disability_recovery}
\end{align}
By assumption, the trend component is $Y(t) = \delta\cdot t$, and we may define 
    \begin{align*}
    \widehat{W}(t) &= \frac{1}{\sigma}\cdot \left(W(t) + Y(t) \right) = \frac{1}{\sigma}\cdot\left(W(t) + \delta \cdot  t\right), \\
    \widehat{M}(T) &= \frac{1}{\sigma}\cdot  \max_{0\leq t\leq T} \left\{ W(t) +  Y(t)\right\} = \max_{0\leq t\leq T} \left\{ \widehat{W}(t)\right\},
\end{align*}
where $\widehat{W}(t)$ is a Brownian motion with drift $\delta/\sigma$, and $\sigma \cdot \widehat{M}(t)$ tracks the maximum vitality depletion until $t$. If $\sigma \cdot \widehat{M}(t)$ is larger than the initial vitality level, then the individual is no longer alive at time $t$. Denote $f(m,w;T)$ as the joint density of $(\widehat{M}(T), \widehat{W}(T))$. Using results from \cite{shreve2004stochastic}, one can derive 
\begin{align*}
   f(m,w;T) = \frac{2(2m-w)}{T\sqrt{2\pi T}} e^{\frac{\delta w}{\sigma} - \frac{\delta^2T}{2\sigma^2} - \frac{(2m-w)^2}{2T}}, w\leq m, m\geq 0,
\end{align*}
and is zero otherwise. 
Notice that the numerator of the expression inside equation \eqref{eq:vitality_disability_recovery} can be written as a function of $(\widehat{M}(T), \widehat{W}(T))$:
\begin{align*}
    &\Pr\left( \left\{v - Y(T)-W(T)>\psi \right\}\text{ and } \left\{v-Y(t)-W(t) > 0   \text{ for any } t\in [0,T] \right\} \right)\\
    =& \Pr \left( v - \sigma \cdot \widehat{W}(T)>\psi,  v - \sigma \cdot \widehat{M}(t) > 0 \right)\\
    =& \Pr \left(    \widehat{W}(T) < \frac{v-\psi}{\sigma}, \ \widehat{M}(T)<\frac{v}{\sigma}\right).
\end{align*}
Therefore, substitute the density function $f(m,v;T)$, we have the recovery probability as:
\begin{align*}
\int_0^{\psi}\frac{1}{F_0(\psi)} \cdot \left(\int_0^{\frac{v}{\sigma}} \int_0^{\frac{v-\psi}{\sigma} } f(m,w; T) \diff w \diff m \right)\diff F_0(v).
\end{align*}
A similarly procedure follows for the disability probability and therefore we omit the details.

\end{proof}

\section{Approximation Methods} \label{appendix: approximation}
\subsection{Tangent Approximation}
If we set $J(t) = 0$, and given the initial vitality level $V(0) = v$,  then \cite{jennen1981first} discussed a tangent approximation method to get the density of $\tau(v)$; that is,
\begin{align*}
\Pr (\tau(v) \in \diff t) = \frac{|H(t,v)-t \cdot H_t(t,v)|}{\sigma \sqrt{2 \pi t^{3}}} e^{-\frac{\left(H(t,v)\right)^{2}}{2 \sigma^{2} t}},
\end{align*}
where $H(t,v):= v - Y(t)$.  

\subsection{Series Expansion}
The tangent approximation has the advantage of simplicity but may sacrifice the accuracy of the approximation. An extension that includes a series expansion has been proposed by \cite{durbin1992first}. The density of $\tau(v)$ is approximated as
\begin{align*}
 \Pr(\tau(v) \in \diff t) \approx \sum_{j=1}^{k}(-1)^{j-1} q_{j}(t), \quad k=1,2, \cdots,
\end{align*}
where 
\begin{align*}
q_{j}(t)= & \int_{0}^{t} \int_{0}^{t_{1}} \cdots \int_{0}^{t_{j-2}}\left[\frac{H\left(t_{j-1},v\right)}{\sigma t_{j-1}}- \sigma^{-1} H_t\left(t_{j-1},v\right)\right] \\ & \times \prod_{i=1}^{j-1}\left[\frac{H\left(t_{i-1},v\right)-H\left(t_{i},v\right)}{\sigma(t_{i-1}-t_{i})}-\sigma^{-1} H_t\left(t_{i-1},v\right)\right] f\left(t_{j-1}, \cdots, t_{1}, t\right) \diff t_{j-1} \cdots \diff t_{1},
\end{align*}
with $f\left(t_{j-1}, \cdots, t_{1}, t\right)$ being the joint density of $(B\left(t_{j-1}\right), \cdots, B\left(t_{1}\right), B(t))$ evaluated at the values of $(H(t_{j-1}, v)/\sigma, \cdots, H(t_{1}, v)/\sigma, H(t,v)/\sigma )$. 
One notable advantage of this approximation method is its rapid convergence when $H$ is concave in $t$. \cite{durbin1992first} demonstrated the accuracy of this series expansion using specific concave functions for $H$, and found that with a choice of $k=3$ (which leads to a double integration), the error lies within $0.001\%$ of the true value.

\section{Cause-of-Death Analysis}

\begin{proof}[Proof of Theorem \ref{decompose tau}]
	We begin by proving equation~\eqref{cpp: LT jump}. Assume that the death occurs at the $n$th ($n\ge 1$) jump in the vitality level (at time $T_n$). Using the independence assumption between $N(t)$ and $\left\{Z_i\right\}_{i\ge 1}$, and conditional on the initial vitality $v$, we have
    \begin{align}\label{cpp: eq1}
 \notag  \mathbb{E}_v\left[\mathrm{e}^{-q\tau_{J}}\mathbb{1}\left(\tau_{J}<\tau_{Y,}\tau_{J}=T_{n}\right)\right]=&\int_{-\infty}^{0}\int_{0}^{\infty}\mathbb{E}_v\left[\mathrm{e}^{-qT_{n}}\mathbb{1}\left(V\left(T_{n}^{-}\right)\in\mathrm{d}z,V\left(T_{n}\right)\in\mathrm{d}y\right) \right]\\
   =&\int_{0}^{\infty}\bar{F}_Z(z) \mathbb{E}_v\left[\mathrm{e}^{-qT_{n}}\mathbb{1}\left(V\left(T_{n}^{-}\right)\in\mathrm{d}z\right) \right].
    \end{align}
Further conditioning on $T_n$, we have
\begin{align*}
\mathbb{E}_{v}\left[\mathrm{e}^{-qT_{n}}\mathbb{1}\left(V\left(T_{n}^{-}\right)\in\mathrm{d}z\right)\right]&=\int_{0}^{\infty}\mathbb{1}(v-Y(t)>z)\mathbb{E}_{v}\left[\mathrm{e}^{-qt}\mathbb{1}\left(V\left(t\right)\in\mathrm{d}z\right)\right]\frac{\lambda^{n}t^{n-1}\mathrm{e}^{-\lambda t}}{(n-1)!}\mathrm{d}t\\
&=\int_{0}^{\infty}\mathbb{1}(v-Y(t)>z)F^{*(n-1)}\left(v-Y(t)-\mathrm{d}z\right)\frac{\lambda^{n}t^{n-1}\mathrm{e}^{-\left(q+\lambda\right)t}}{(n-1)!}\mathrm{d}t.
\end{align*}
Substituting the above equation into equation~\eqref{cpp: eq1} yields
\begin{align}\label{cpp: eq1-1}
\notag &	 \mathbb{E}_v\left[\mathrm{e}^{-q\tau_{J}}\mathbb{1}\left(\tau_{J}<\tau_{Y,}\tau_{J}=T_{n}\right)\right]\\
	\notag =&\int_{0}^{\infty}\mathbb{1}\left(v-Y(t)>0\right)\left(\int_{0}^{v-Y(t)}f_{Z}^{*(n-1)}(v-Y(t)-z)\bar{F}_{Z}(z)\mathrm{d}z\right)\frac{\lambda^{n}t^{n-1} \mathrm{e}^{-(q+\lambda) t}}{(n-1)!}\mathrm{d}t\\
\notag	=&\int_{0}^{\infty}\mathbb{1}\left(v-Y(t)>0\right)\left(F_{Z}^{*(n-1)}\left(v-Y(t)\right)-F_{Z}^{*(n)}\left(v-Y(t)\right)\right)\frac{\lambda^{n}t^{n-1} \mathrm{e}^{-(q+\lambda) t}}{(n-1)!}\mathrm{d}t\\
	\notag       =&\int_{0}^{\infty}\mathbb{1}(v-\frac{bc^{x}}{\ln c}\left(c^{t}-1\right)>0)\frac{\mathrm{e}^{-\alpha\left(v-\frac{bc^{x}}{\ln c}\left(c^{t}-1\right)\right)}\alpha^{n-1}\left(v-\frac{bc^{x}}{\ln c}\left(c^{t}-1\right)\right)^{n-1}}{(n-1)!}\frac{\lambda^{n}t^{n-1} \mathrm{e}^{-(q+\lambda)t}}{(n-1)!}\mathrm{d}t\\
	=& \int_{0}^{t_v^{*}}\frac{\lambda \mathrm{e}^{-\alpha\left(v-\frac{bc^{x}}{\ln c}\left(c^{t}-1\right)\right)-(q+\lambda)t}\left(\lambda\alpha vt-\frac{\lambda\alpha bc^{x}t}{\ln c}\left(c^{t}-1\right)\right)^{n-1}}{(n-1)!(n-1)!}\mathrm{d}t,
\end{align}
 where the third equation follows from the fact that the $n$-fold convolution of $F_Z(z)$ with itself is given by
\begin{align*}
	F_{Z}^{*(n)}(z)&=1-\sum_{i=0}^{n-1}\frac{1}{i!}\mathrm{e}^{-\alpha z}(\alpha z)^{i}.
\end{align*}  
Summing up equation~\eqref{cpp: eq1-1} over $n$ leads to completes the proof of equation~\eqref{cpp: LT jump}.

Now, suppose that the death occurs due to the decrement in $Y(\cdot)$ rather than a fatal accident, i.e., $\tau = \tau_Y < \tau_J$. Conditioning on $\{T_n=t, V(T_n) = z>0\}$, the event that there are $n$ ($n \ge 0$) accidents occurred before death is equivalent to the event $\{\tau_Y=t+s, T_{n+1} - T_n > s\}$, where $s$ is the solution to $Y(t+s)-Y(t)=z$ (i.e., the vitality level $V$ decreases to $0$ before the $(n+1)$th accident occurs). Solving the equation in $s$ leads to
\begin{align*}
s = \frac{\ln\left(bc^{x+t}+z\ln c\right)-\ln\left(bc^{x+t}\right)}{\ln c}.
\end{align*}
Along with the fact that $T_{n+1} - T_n$ is exponentially distributed and independent of $(T_n, V(T_n))$, we have
\begin{align}
 \notag   &\mathbb{E}_v\left[\mathrm{e}^{-q\tau_{Y}}\mathbb{1}\left(\tau_{Y}<\tau_{J,}N(\tau_{Y})=n\right)\right]\\
\notag =&\int_{0}^{\infty}\mathbb{1}\left(v-Y(t)>0\right)\int_{0}^{v-Y(t)}\mathrm{e}^{-q(t+s)}\mathrm{Pr}_v\left(V\left(T_{n}\right)\in\mathrm{d}z,T_{n}\in\mathrm{d}t,T_{n+1}-T_{n}>s\right)\\
\notag =&\int_{0}^{\infty}\mathbb{1}\left(v-Y(t)>0\right)\int_{0}^{v-Y(t)}\mathrm{e}^{-(q+\lambda)s}f_{Z}^{*(n)}(v-Y(t)-z)\frac{\lambda^{n}t^{n-1} \mathrm{e}^{-(q+\lambda) t}}{(n-1)!}\mathrm{d}z\mathrm{d}t.
\end{align}
By substituting the representation of $s$ and the form of $Y$ into the above equation and performing some direct calculations, we obtain
\begin{align} \label{cpp: eq4}
\notag    &\mathbb{E}_v\left[\mathrm{e}^{-q\tau_{Y}}\mathbb{1}\left(\tau_{Y}<\tau_{J,}N(\tau_{Y})=n\right)\right]\\
\notag =&\int_{0}^{t_v^{*}}\int_{0}^{v-\frac{bc^{x+t}-bc^{x}}{\ln c}}\left(\frac{\left(v-z\right)\ln c+bc^{x}}{bc^{x+t}}\right)^{-\frac{q+\lambda}{\ln c}}\frac{\alpha^{n}z^{n-1} \mathrm{e}^{-\alpha z}}{(n-1)!}\mathrm{d}z\frac{\lambda^{n}t^{n-1} \mathrm{e}^{-(q+\lambda)t}}{(n-1)!}\mathrm{d}t\\
 \notag=&\int_{0}^{t_v^{*}}\int_{0}^{v-\frac{bc^{x+t}-bc^{x}}{\ln c}}\left(\frac{\left(v-z\right)\ln c+bc^{x}}{bc^{x}}\right)^{-\frac{q+\lambda}{\ln c}}\frac{\alpha^{n}z^{n-1} \mathrm{e}^{-\alpha z}}{(n-1)!}\mathrm{d}z\frac{\lambda^{n}t^{n-1}}{(n-1)!}\mathrm{d}t\\
\notag =&\int_{0}^{v}\left(\frac{\ln\frac{\left(v-z\right)\ln c+bc^{x}}{bc^{x}}}{\ln c}\right)^{n}\frac{\lambda^{n}}{n!}\left(\frac{(v-z)\ln c+bc^{x}}{bc^{x}}\right)^{-\frac{q+\lambda}{\ln c}}\frac{\alpha^{n}z^{n-1}\mathrm{e}^{-\alpha z}}{(n-1)!}\mathrm{d}z\\
 =&\int_{0}^{t_v^{*}}\mathrm{e}^{-\left(q+\lambda\right)t-\alpha b_v^*}\frac{\left(b_v^*\right)^{n-1}\left(\alpha\lambda t\right)^{n}}{n!(n-1)!}bc^{x+t}\mathrm{d}t,
\end{align}
for $n\ge 1$, where the third equation follows from interchanging the order of integrations, while the last equation is derived by changing of variable (letting $\frac{\ln\left(bc^{x}+z\ln c\right)-\ln b}{\ln c}-x=t$).~Summing up equation~\eqref{cpp: eq4} over $n$ ($n\ge 1$) and noting that 
\begin{align*}
\mathbb{E}_v\left[\mathrm{e}^{-q\tau_{Y}}\mathbb{1}\left(\tau_{Y}<\tau_{J,}N(\tau_{Y})= 0\right)\right] = \mathrm{e}^{-(q+\lambda)t_v^{*}}
\end{align*}
completes that proof of equation~\eqref{cpp: LT creep}.

\end{proof}

\end{document}